\documentclass[a4paper,10pt]{article}

\usepackage{amssymb}
\usepackage[leftmargin=1em,vskip=0.5em]{quoting}
\usepackage{amsmath}
\usepackage{amsthm}
\usepackage{url}
\usepackage{tikz}
\usepackage{color} 
\usepackage{graphics}
\usepackage{enumerate}
\usepackage{microtype}
\usepackage{subcaption}
\usepackage{authblk}

\usepackage{thmtools,thm-restate}

\usepackage[boxed, ruled,noend,linesnumbered]{algorithm2e}

\usepackage{wrapfig}
\usepackage{cite}

\usepackage[bottom=1in,top=1in,left=1in,right=1in]{geometry}

\newtheorem{theorem}{Theorem}
\newtheorem{corollary}[theorem]{Corollary}
\newtheorem{definition}{Definition}
\newtheorem{lemma}[theorem]{Lemma}
\newtheorem{property}[theorem]{Property}

\def\Car{\ensuremath{\mathit{Car}}}
\def\HSS{\ensuremath{\mathit{csize}}}
\def\setcon{\ensuremath{\mathit{setcon}}}
\def\facet{\ensuremath{\mathit{facet}}}

\def\facets{\ensuremath{\mathit{facets}}}
\def\faces{\ensuremath{\mathit{faces}}}
\def\Sub{\ensuremath{\mathit{Sub}}}
\def\Cl{\ensuremath{\mathit{Cl}}}
\def\St{\ensuremath{\mathit{St}}}
\def\Pc{\ensuremath{\mathit{Pc}}}

\newcommand{\myparagraph}[1]{\vspace{6pt}\noindent \textbf{#1}}

\newcommand{\remove}[1]{}
\newcommand{\ignore}[1]{}

\def\A{\ensuremath{\mathcal{A}}}

\def\R{\ensuremath{\mathcal{R}}}

\def\I{\ensuremath{\mathcal{I}}}
\def\O{\ensuremath{\mathcal{O}}}

\def\C{\ensuremath{\mathcal{C}}}

\def\fair{\textit{fair}}

\def\HSS{\mathit{csize}}

\def\s {\mathbf{s}}
\def\t {\mathbf{t}}

\def\Chr{\operatorname{Chr}}

\def\O {\mathcal{O}}
\def\I {\mathcal{I}}

\def\Skel {\operatorname{Skel}}
\def\x{\mathbf{x}}

\def\Car{\mathit{carrier}}

\def\IStwo{\ensuremath{\mathit{IS2}}}
\def\ISone{\ensuremath{\mathit{IS1}}}
\def\CS{\ensuremath{\mathcal{CS}_\alpha}}
\def\CSM{\ensuremath{\mathcal{CSM}_\alpha}}
\def\CSV{\ensuremath{\mathcal{CSV}_\alpha}}

\def\IS{\textit{IS}}

\renewcommand\th{^{\text{th}}}

\begin{document}

\title{An Asynchronous Computability Theorem for Fair Adversaries\footnote{This work has been supported by the Franco-German DFG-ANR Project DISCMAT (14-CE35-0010-02) devoted to connections between mathematics and distributed computing.}}

\author[1]{Petr Kuznetsov}
\author[1]{Thibault Rieutord}
\author[2]{Yuan He}
\affil{LTCI, T\'el\'ecom ParisTech, Universit\'e Paris-Saclay, Paris, France\\
  \texttt{\{petr.kuznetsov,thibault.rieutord\}@telecom-paristech.fr}}
\affil[2]{UCLA, Los Angeles, USA\\
 \texttt{yuan.he@cs.ucla.edu}}

\date{}
\maketitle
\begin{abstract}
This paper proposes a simple topological characterization of
a large class of \emph{fair} adversarial models via \emph{affine tasks}: 
sub-complexes of the second iteration of the standard chromatic subdivision.
We show that the task computability of a model in the class 
is precisely captured by iterations of the corresponding affine task.
Fair adversaries include, but are not restricted to, the models of
wait-freedom, $t$-resilience, and $k$-concurrency.  
Our results generalize and improve all previously derived
topological characterizations of the ability of a model to solve
distributed tasks.  

\end{abstract}

%

\section{Introduction}

Distributed computing deals with a jungle of models, parameterized by 
types of failures, synchrony assumptions and communication primitives.
Determining relative computability power of these models (``is model
$A$ more powerful than model $B$?'') is an intriguing and important
problem.

This paper deals with \emph{shared-memory} models in which a set of
\emph{crash-prone} \emph{asynchronous} processes communicate
via invoking operations on a collection of shared objects, 
which, by default, include atomic read-write registers.

\paragraph*{Topology of wait-freedom.} 
The \emph{wait-free} model~\cite{Her91} 
makes no assumptions on when and where failures might occur.
Herlihy and Shavit proposed an elegant 
characterization of wait-free (read-write) task computability via the
existence of a specific \emph{simplicial} map from geometrical
structures describing inputs and outputs~\cite{HS99}.

A task $T$ has a wait-free solution using read-write registers  
if and only if there exists a simplicial, chromatic map from some
\emph{subdivision}  of the \emph{input simplicial complex}, describing
the inputs of $T$, to the \emph{output simplicial complex}, describing
the outputs of $T$, respecting the task specification $\Delta$.
In particular, we can choose this subdivision to be a number of
iterations of the \emph{standard chromatic} subdivision (denoted~$\Chr$, Figure~\ref{Fig:scs}).

Therefore, the celebrated \emph{Asynchronous Computability Theorem (ACT)}~\cite{HS99}
can be formulated~as:
\begin{quoting}
A task $T=(\I,\O,\Delta)$, where $\I$ is the input complex, $\O$ is an
output complex, and $\Delta$ is a carrier map from $\I$ to sub-complexes of
$\O$, is wait-free solvable 
if and only if there exists a natural number $\ell$ and a simplicial map~
$\phi: \Chr^{\ell}(\I) \rightarrow \O$ carried by $\Delta$
(informally, respecting the task specification $\Delta$).
\end{quoting}

\begin{figure}
\captionsetup[subfigure]{justification=centering}
  \begin{minipage}[b]{.49\linewidth}
    \centering
\includegraphics[scale=0.54]{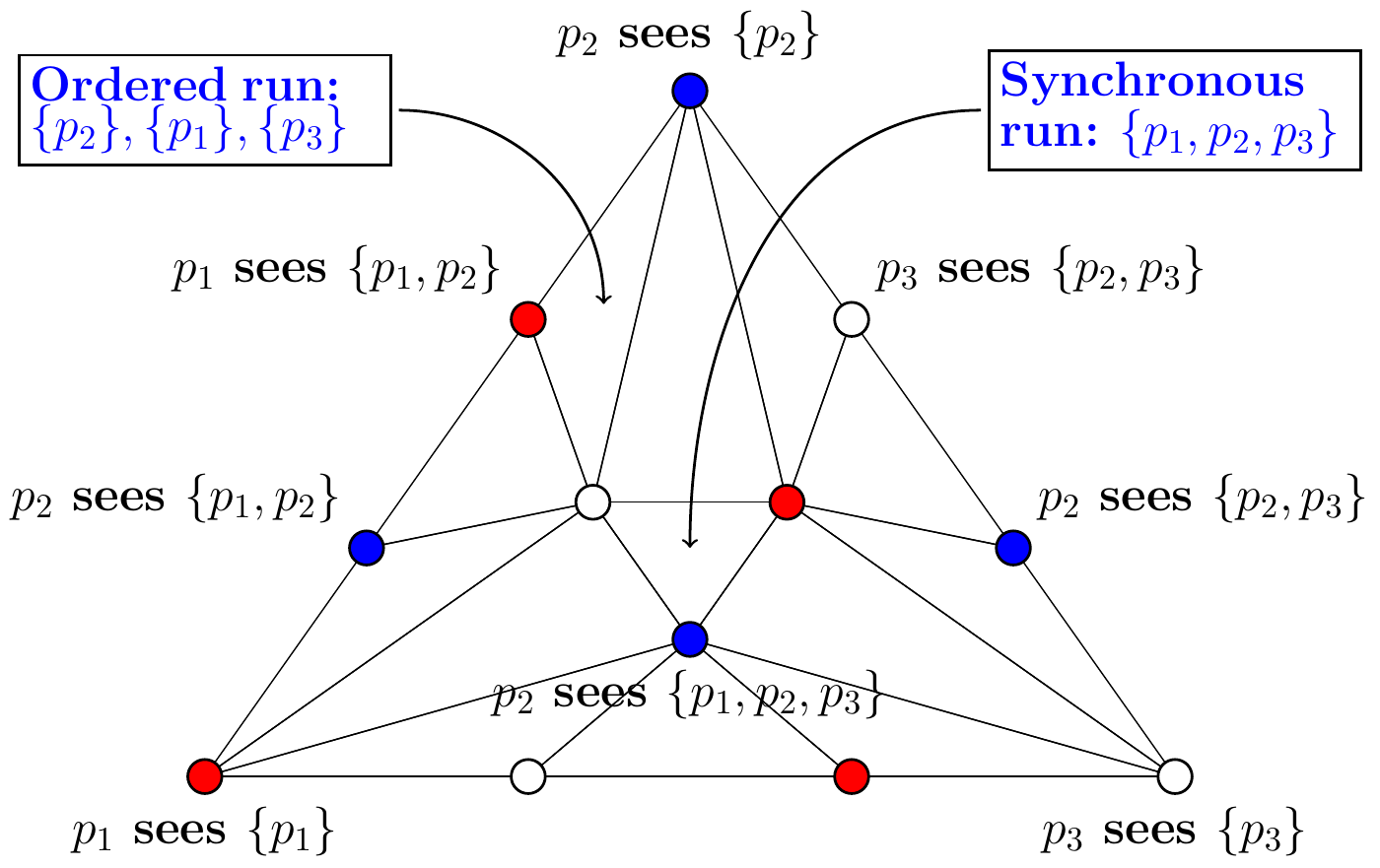}
   \subcaption{\label{Fig:scs}$\Chr(\s)$, the standard chromatic
     subdivision of a $2$-simplex, the output complex of the $3$-process $\IS$
     task.}
  \end{minipage}
  \hfill
  \begin{minipage}[b]{.49\linewidth}
    \begin{center}
\includegraphics[scale=0.59]{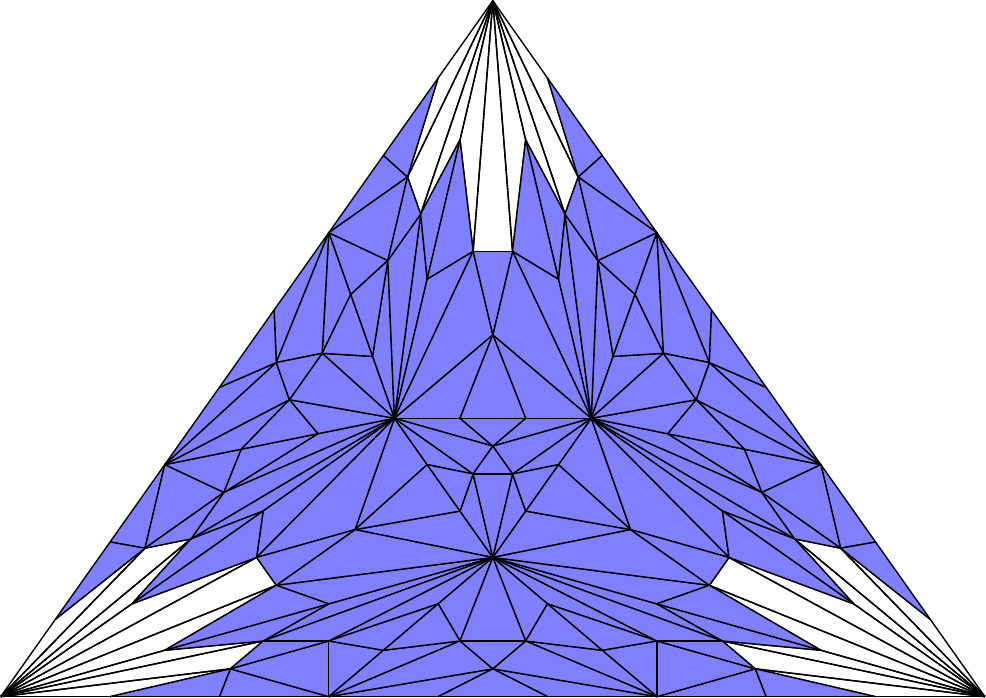}
     \subcaption{\label{Fig:R-1}In blue, $\R_{1-\mathit{res}}$, the affine task of $1$-resilience for a $3$-process system.
     }
    \end{center}
  \end{minipage}
  \caption{\small{Standard chromatic subdivision and affine task example.}}
\end{figure}

The output complex of the \emph{immediate snapshot} (IS) task is 
precisely captured by $\Chr$~\cite{BG97}. By solving the $\IS$ task iteratively, 
where the current iteration output is used as the input value for the next one, 
we obtain the \emph{iterated} immediate snapshot (IIS) model, 
captured by iterations of $\Chr$.
The ACT theorem can thus be interpreted as: the set of wait-free (read-write) 
solvable task is precisely the set of tasks solvable in the IIS model. 
The ability of (iteratively) solving the $\IS$ task thus allows us to solve
any task in the wait-free model.
Hence, from the task computability perspective, the $\IS$ task is a finite 
representation of the wait-free model.  

\paragraph*{Adversaries.}
Given that many fundamental tasks are not solvable in the wait-free
way~\cite{BG93b,HS99,SZ00}, more general models were considered.  
Delporte et al.~\cite{DFGT11}  introduced the notion of an \emph{adversary},  
a collection~$\A$ of process subsets, called \emph{live sets}.
A run is in the corresponding \emph{adversarial $\A$-model} if the set of
processes taking infinitely many steps in it is a live set of $\A$.
For example, the \emph{$t$-resilient} $n$-process model is defined via
an adversary $\A_{t-\mathit{res}}$ that consists
of all process subsets of size $n-t$ or more.
$\A_{t-\mathit{res}}$ is \emph{superset-closed}~\cite{Kuz12},
as it contains all supersets of its elements.

Saraph et al.~\cite{SHG16} recently proposed a direct characterization
of $t$-resilient task computability via a specific task~$\R_{t-\mathit{res}}$.
The task is defined as a restriction of the \emph{double} immediate snapshot
task: the output complex of the task is a sub-complex consisting of
\emph{all} simplices of the second degree of the standard chromatic
subdivision of the task's input complex, except the simplices 
adjacent to the $(n-t-1)$-skeleton of the input complex. 
Intuitively the output complex of $\mathcal{R}_{t-\mathit{res}}$ contains all of 
$2$-round $\IS$ runs in which every process ``sees'' at least $n-t-1$ other processes.
We call such tasks \emph{affine}~\cite{GKM14,GHKR16}, as
the geometrical representation of their output complexes are
unions of affine spaces.  
An affine task consists in solving a (generalized) simplex
agreement~\cite{BG97,HS99} on the corresponding sub-complex of~$\Chr^2\s$.

Figure~\ref{Fig:R-1} depicts the output complex of $\R_{1-\mathit{res}}$,
the affine task for the $3$-process $1$-resilient model. 

Solving a task $T$ in the $t$-resilient model is then equivalent 
to finding a map from iterations of $\R_{t-\mathit{res}}$ 
(applied to the input complex of $T$) to the output complex of $T$.   

Similarly, the affine task of the \emph{$k$-obstruction-free} adversary, consisting of all
process subsets of  size at most $k$, was recently determined by Gafni
et al.~\cite{GHKR16}. Note that such an adversary is
\emph{symmetric}~\cite{Tau10}, 
as it only depends on the \emph{sizes}
of live sets, and not on process identifiers.
Unlike $\A_{t-\mathit{res}}$ (which is also symmetric), the
$k$-obstruction-free one is not superset-closed.

\paragraph*{Topology of fair adversaries.}
In this paper, we present a compact topological characterization
of the large class of \emph{\fair} adversarial models~\cite{KR17}.
Informally, an adversary is fair if a subset of the participating processes $P$
cannot achieve better set consensus than $P$.
Fair adversaries subsume, but are not restricted to, symmetric and superset-closed ones.

We define an affine task $\R_{\A}$ capturing the task computability 
of any fair (adversarial) $\A$-model.  
Our characterization can be put as the following generalization of the
ACT~\cite{HS99}:
\begin{quoting}
A task $T=(\I,\O,\Delta)$ is solvable in a fair adversarial $\A$-model
if and only if there exists a natural number $\ell$ and a simplicial map
$\phi: \R_{\A}^{\ell}(\I) \rightarrow \O$ carried by $\Delta$.
\end{quoting}

This result generalizes all existing topological characterizations of
distributed computing models~\cite{HS99,GKM14,GHKR16,SHG16},
as it applies to \emph{all} fair adversaries (and not only
$t$-resilient and $k$-obstruction-free) and \emph{all} tasks (and not
only \emph{colorless}).

Figure~\ref{Fig:related} shows adversary classes
and summarizes the results of this paper along with earlier affine characterizations.

We believe that the results can be extended to
all ``practical'' restrictions of the wait-free model, beyond fair
adversaries, which may result in a complete computability
theory for distributed computing shared-memory models.

\paragraph*{Compact models.}
Our affine-task formalism  enables a \emph{compact} representation of
a distributed computing model.     
Intuitively, assuming the conventional ``longest-prefix'' metric~\cite{AS85}, 
a model $M$, as a set of infinite runs, is compact if it contains its limit points: 
if every prefix of an infinite run \emph{complies} with~$M$ 
(i.e., can be extended to a run in $M$), then the run is in $M$. 
$M$ can then be viewed as a
\emph{safety property}~\cite[Chap.~8]{Lyn96}: 
to check if a run is in $M$, it is sufficient to check whether each
of its finite prefixes complies with $M$.

Most adversarial models are non-compact.
For example, the $1$-obstruction-free $2$-process model is compliant with all finite runs,
but among the infinite ones---only those in which exactly one process runs solo from
some point on are in the model.
Similarly, consider an infinite solo run in which exactly one process takes steps.
All finite prefixes of this run complies with the $1$-resilient $3$-process model,
but the run itself is not in the model. 
 
In contrast, the affine model $L^*$, defined as the subset of infinite IIS
runs resulting by iterating an affine task $L$ is, by construction, compact.
By a simple application of K\"onig's lemma, we can easily show that every task
solvable in an affine model is solvable in a \emph{bounded} number of
IIS rounds, i.e., in finitely many finite runs.
In a non-compact model, such as the model of $1$-resilience, many tasks can only be solved
in arbitrarily long runs, hence, to check if a solution is correct, we
might have to consider infinitely many infinite runs.
Thus, working in an equivalent affine model may be attractive from the
verification viewpoint.  

\begin{figure}
\center
\includegraphics[scale=1.2]{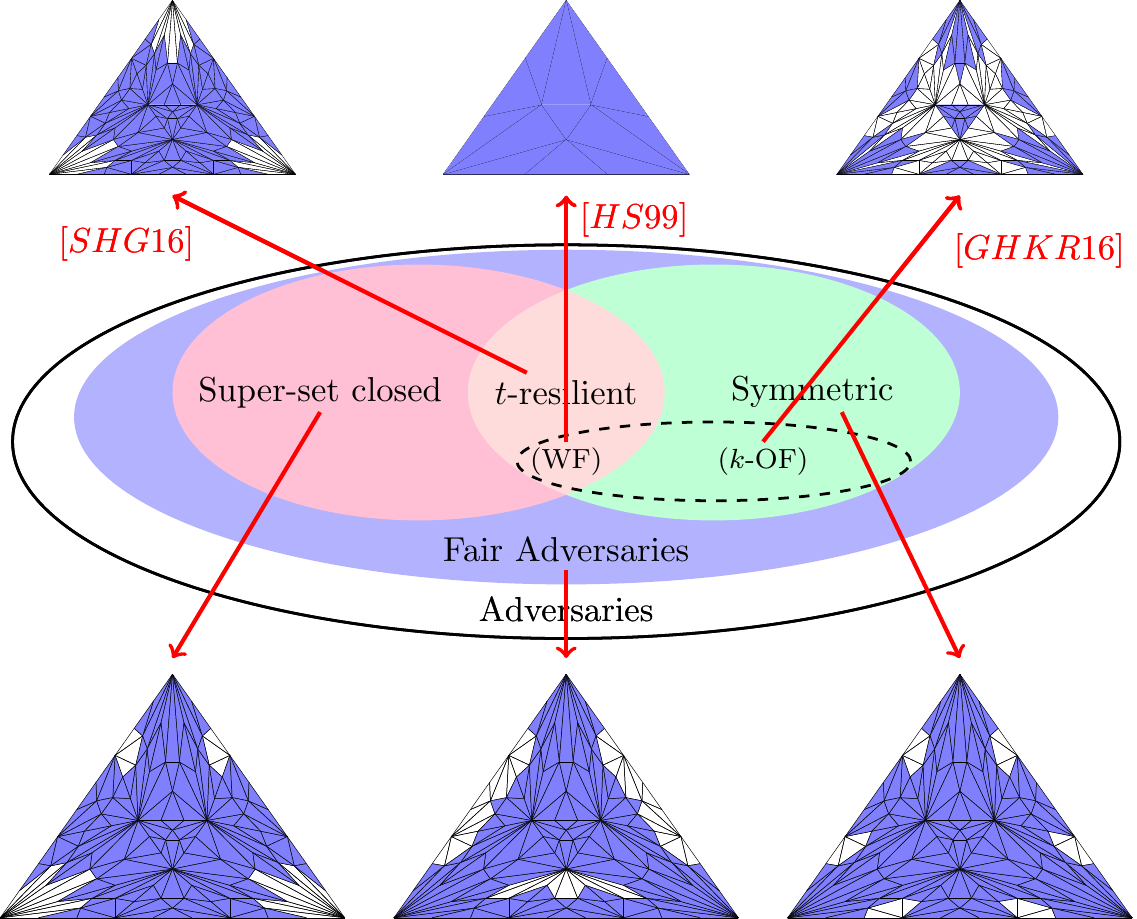}
\caption{\label{Fig:related} Earlier topological
  characterizations of the wait-free~\cite{HS99},
$t$-resilient~\cite{SHG16} and $k$-obstruction-free~\cite{GHKR16} models,
and our contribution: affine tasks for all fair adversaries.} 
\end{figure}

\paragraph*{Roadmap.} 
Section~\ref{sec:model} describes our model.  
Section~\ref{sec:adv} recalls the definitions of adversarial models and 
presents agreement functions.    
Section~\ref{sec:def} defines the affine task $\R_{\A}$ for 
a fair adversary $\A$.
In Section~\ref{sec:algo}, we show that $\R_{\A}^*$ 
can be simulated in the adversarial $\A$-model. 
In Section~\ref{sec:alpha}, we show that any task solvable in 
the $\A$-model can be solved in~$\R_{\A}^*$.
Section~\ref{sec:related} reviews related work and
Section~\ref{sec:disc} concludes the paper.

\section{Preliminaries}
\label{sec:model}

We assume a system of $n$ asynchronous processes, 
$\Pi=\{p_1,\ldots,p_n\}$. Two models of communication are considered: 
(1)~\emph{atomic snapshots}~\cite{AADGMS93} and 
(2)~\emph{iterated immediate snapshots}~\cite{BG97,HS99}.

\paragraph*{Atomic snapshot models.}
The atomic-snapshot (AS) memory is represented as a vector of $n$ 
shared variables, where each process $p_i$ is associated with the 
position $i$. The memory can be accessed with two operations: \emph{update} and 
\emph{snapshot}. An \emph{update} operation performed by $p_i$ 
modifies the value at position $i$ 
and a \emph{snapshot} returns the vector current state.

A \emph{protocol} is a deterministic distributed automaton that, for each 
process and each its local state, stipulates which operation and 
state transition the process may perform. 
A \emph{run} of a protocol is a possibly infinite 
sequence of alternating states and operations.
An AS \emph{model} is a set of infinite runs.

In an infinite run of the AS model,
a process that takes only  finitely many steps
is called \emph{faulty}, otherwise it is called 
\emph{correct}. We assume that in its first step, a process 
shares its initial state using the \emph{update} operation. 
If a process completed this first step in a given run, it is said to be
\emph{participating}, and the set of participating processes is called
the \emph{participating set}.
Note that every correct process is participating.

\paragraph*{Iterated Immediate Snapshot model.}
In the iterated immediate snapshot (IIS) model,
processes proceed through an infinite sequence of independent 
memories~$M_1, M_2,\ldots$. Each memory~$M_r$ is accessed 
by a process with a single \emph{WriteSnapshot} 
operation~\cite{BG93a}: the operation performed by~$p_i$ takes a
value~$v_{ir}$ and returns a set~$V_{ir}$ of submitted values 
(w.l.o.g, values of different processes are distinct), 
satisfying the following properties (See Figure~\ref{fig:ISexamples} for IS examples): 
\begin{itemize}
\item self-inclusion: $v_{ir} \in V_{ir}$; 
\item containment: $(V_{ir}\subseteq V_{jr}) \vee (V_{jr}\subseteq V_{ir})$;
\item immediacy: $v_{ir} \in V_{jr}$ $\Rightarrow$ $V_{ir}\subseteq V_{jr}$.
\end{itemize}

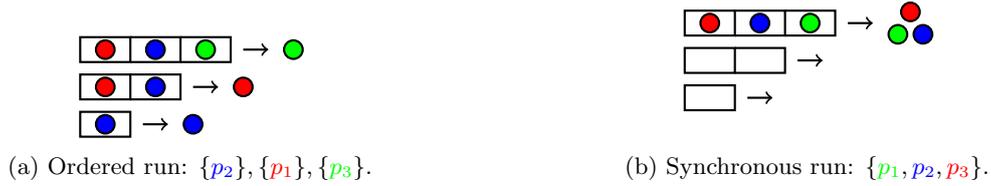
\begin{figure}
\captionsetup[subfigure]{justification=centering}
\begin{minipage}[b]{0.49\linewidth}
\begin{center}
    \begin{tikzpicture}[scale=0.33]
	\draw[thick] (0,5) -- (2,5) -- (2,4) -- (0,4) -- cycle;
	\draw (1,4.5) node[circle,inner sep=0pt,minimum size=7, fill=red,draw,thick] {};
	\draw[thick] (2,5) -- (4,5) -- (4,4) -- (2,4);
	\draw (3,4.5) node[circle,inner sep=0pt,minimum size=7, fill=blue,draw,thick] {};
	\draw[thick] (4,5) -- (6,5) -- (6,4) -- (4,4);
	\draw (5,4.5) node[circle,inner sep=0pt,minimum size=7, fill=green,draw,thick] {};
	\draw[->, thick] (6.5,4.5) -- (7.5,4.5);
	\draw (8.5,4.5) node[circle,inner sep=0pt,minimum size=7, fill=green,draw,thick] {};
	
	\draw[thick] (0,3.5) -- (2,3.5) -- (2,2.5) -- (0,2.5) -- cycle;
	\draw (1,3) node[circle,inner sep=0pt,minimum size=7, fill=red,draw,thick] {};
	\draw[thick] (2,3.5) -- (4,3.5) -- (4,2.5) -- (2,2.5);
	\draw (3,3) node[circle,inner sep=0pt,minimum size=7, fill=blue,draw,thick] {};
	\draw[->, thick] (4.5,3) -- (5.5,3);
	\draw (6.5,3) node[circle,inner sep=0pt,minimum size=7, fill=red,draw,thick] {};
	
	\draw[thick] (0,2) -- (2,2) -- (2,1) -- (0,1) -- cycle;
	\draw (1,1.5) node[circle,inner sep=0pt,minimum size=7, fill=blue,draw,thick] {};
	\draw[->, thick] (2.5,1.5) -- (3.5,1.5);
	\draw (4.5,1.5) node[circle,inner sep=0pt,minimum size=7, fill=blue,draw,thick] {} ;
	
	\end{tikzpicture}
	\subcaption{\small Ordered run: $\{\color{blue}p_2\color{black}\},\{\color{red}p_1\color{black}\},\{\color{green}p_3\color{black}\}$.}
    \end{center}
\end{minipage}
    \hfill
    \begin{minipage}[b]{0.49\linewidth}
    \begin{center}
\begin{tikzpicture}[scale=0.33]
	\draw[thick] (0,5) -- (2,5) -- (2,4) -- (0,4) -- cycle;
	\draw (1,4.5) node[circle,inner sep=0pt,minimum size=7, fill=red,draw,thick] {};
	\draw[thick] (2,5) -- (4,5) -- (4,4) -- (2,4);
	\draw (3,4.5) node[circle,inner sep=0pt,minimum size=7, fill=blue,draw,thick] {};
	\draw[thick] (4,5) -- (6,5) -- (6,4) -- (4,4);
	\draw (5,4.5) node[circle,inner sep=0pt,minimum size=7, fill=green,draw,thick] {};
	\draw[->, thick] (6.5,4.5) -- (7.5,4.5);
	\draw (8.5,4.05) node[circle,inner sep=0pt,minimum size=7, fill=green,draw,thick] {};
	\draw (9.5,4.05) node[circle,inner sep=0pt,minimum size=7, fill=blue,draw,thick] {} ;
	\draw (9,4.95) node[circle,inner sep=0pt,minimum size=7, fill=red,draw,thick] {};
	
	\draw[thick] (0,3.5) -- (2,3.5) -- (2,2.5) -- (0,2.5) -- cycle;
	\draw[thick] (2,3.5) -- (4,3.5) -- (4,2.5) -- (2,2.5);
	\draw[->, thick] (4.5,3) -- (5.5,3);
	
	\draw[thick] (0,2) -- (2,2) -- (2,1) -- (0,1) -- cycle;
	\draw[->, thick] (2.5,1.5) -- (3.5,1.5);
	\end{tikzpicture}
    \end{center}
	\subcaption{\small Synchronous run: $\{\color{green}p_1\color{black},\color{blue}p_2\color{black},\color{red}p_3\color{black}\}$.}
\end{minipage}
\caption{Examples of valid sets of IS outputs.\label{fig:ISexamples}}
\end{figure}

In the IIS communication model, we assume that processes run the 
\emph{full-information} protocol, in which, the first value each process 
writes is its \emph{initial state}. For each~$r>1$, the outcome of 
the WriteSnapshot operation on memory $M_{r-1}$ is submitted as 
the input value for the WriteSnapshot operation on~$M_r$. 
There are no failures in the IIS model,
all processes go through infinitely many~$\IS$ instances. 

Note that the wait-free AS model and the IIS model are equivalent as regards 
task solvability~\cite{BG93a,HKR14}.

\paragraph*{Tasks.}
In this paper, we focus on distributed \emph{tasks}~\cite{HS99}. 
A process invokes a task with an \emph{input} value and the task 
returns an \emph{output} value, so that the inputs and the outputs 
across the processes respect the task specification.
Formally, a \emph{task} is defined through a set~$\I$ of input 
vectors (one input value for each process), a set~$\O$ of output 
vectors (one output value for each process), and a total 
relation~$\Delta:\I\mapsto 2^{\O}$ that associates each input 
vector with a set of possible output vectors.
We require that $\Delta$ is a \emph{carrier} map: $\forall \rho,\sigma
\in \I$, $\rho\subseteq\sigma$: $\Delta(\rho)\subseteq \Delta(\sigma)$.  
An input~$\bot$ denotes a \emph{non-participating} process and an output 
value~$\bot$ denotes an \emph{undecided} process. 
Check~\cite{HKR14} for more details on the definition.

In the \emph{$k$-set consensus} task~\cite{Cha90}, 
input values are in a set of values $V$ ($|V|\geq k+1$), output values
are in~$V$, and for each input vector $I$ and output vector $O$, 
$(I,O) \in\Delta$ if the set of non-$\bot$ values in~$O$ is 
a subset of values in $I$ of size at most $k$. 
The case of $1$-set consensus is called \emph{consensus}~\cite{FLP85}. 

A protocol solves a task $T=(\I,\O,\Delta)$ in a model $M$, 
if it ensures that in every run of~$M$ in which processes 
start with an input vector $I\in\I$, there is a finite
prefix $R$ of the run in which: (1)~decided values form a 
vector $O\in\O$ such that~$(I,O)\in\Delta$, and (2)~all 
correct processes decide. 
Hence, in the IIS model, all processes must decide.

\paragraph*{Simplicial complexes.}
We use the standard language of \emph{simplicial
  complexes}~\cite{Spanier,HKR14} to give a combinatorial
representation of the IIS model.  
A \emph{simplicial complex} is defined as a set of 
\emph{vertices} and an inclusion-closed set of vertex subsets, 
called \emph{simplices}. 
The \emph{dimension} of a simplex $\sigma$ is 
$|\sigma|-1$, and any subset of $\sigma$ is one of its \emph{faces}. 
We denote by~$\s$ the {\em standard $(n-1)$-simplex}: a fixed set of $n$ vertices
and all its subsets. 

Given a complex $K$ and a simplex~$\sigma\in K$, 
$\sigma$ is a \emph{facet} of $K$, 
denoted $\facet(\sigma,K)$, if $\sigma$ is not a face of any stricly larger simplex in $K$.
Let $\facets(K)= \{\sigma\in K,\facet(\sigma,K)\}$.
A simplicial complex is \emph{pure} (of dimension~$n$) if all its facets have 
dimension $n$.   
A simplicial complex is \emph{chromatic} if it is equipped with 
a \emph{coloring function} --- a non-collapsing simplicial map~$\chi$ from its vertices to $\s$, in 
one-to-one correspondence with $n$ {\em colors}. 
In our setting, colors correspond to processes identifiers.

The \emph{closure} of a set of simplices  $S$, $\Cl(S)$, 
is the complex formed by all faces of simplices in $S$, i.e., 
$\bigcup_{\sigma\in S} \faces(\sigma)$.
Given a complex~$K$, the \emph{star} of $S\subseteq K$ in $K$, $\St(S,K)$, 
is the set of all simplices in $K$ having a simplex from $S$ as a face, i.e., 
$\{\sigma\in K|\faces(\sigma)\cap S\neq\emptyset\}$.
Given a pure complex~$K$, we also define a new construct 
that we call the \emph{pure complement} of $S\subseteq K$ in $K$, 
$\Pc(S,K)$. It is the maximal pure sub-complex of~$K$ 
of the same dimension as $K$ which does not intersect with~$S$, i.e., 
$\Cl(\{\sigma\in\facets(K)|\faces(\sigma)\cap S=\emptyset\})$.

\paragraph*{Standard chromatic subdivision and IIS}
The \emph{standard chromatic subdivision}~\cite{HS99} of a complex $K$, $\Chr K$ 
($\Chr\s$ is depicted in Figure~\ref{Fig:scs}), 
is a complex where vertices of $\Chr K$ are couples $(c, \sigma)$,
where $c$ is a color and~$\sigma$ is a face of~$K$ containing a vertex of color~$c$. 
Simplices of $\Chr K$ are the sets of vertices $(c_1,\sigma_1)$,~$\ldots$,~$(c_m,\sigma_m)$  
associated with distinct colors (i.e., $\forall i,j$, $c_i\neq c_j$)
such that the $\sigma_i$ satisfies the containment and 
immediacy properties of $\IS$. 
It has been shown that $\Chr$ is a \emph{subdivision}~\cite{Koz12}. In particular, 
the geometric realization of $\Chr\s$, $|\Chr\s|$, is homeomorphic to $|\s|$, 
the geometric realization of~$\s$ (i.e., the convex hull of its vertices).
If we \emph{iterate} this subdivision~$m$ times, each time 
applying $\Chr$ to all simplices, we 
obtain the~$m^{th}$ chromatic subdivision,~$\Chr^m$.
$\Chr^m \s$ captures the $m$-round IIS model, $\IS^m$~\cite{BG97,HS99}.

Given a complex $K$ and a subdivision of it $\Sub(K)$, 
the carrier of a simplex $\sigma\in \Sub(K)$ in $K$, 
$\Car(\sigma, K)$, is the smallest simplex $\rho\in K$ 
such that the geometric realization of~$\sigma$,~$|\sigma|$, is contained in $|\rho|$: 
$|\sigma|\subseteq|\rho|$. 
The carrier of a vertex~$(p,\sigma)\in\Chr \s$ is $\sigma$. 
In the matching $\IS$ task, the carrier corresponds to the snapshot returned by~$p$, 
i.e., the set of processes \emph{seen} by $p$. The carrier of a simplex 
$\rho \in \Chr K$ is simply the union (or, due to inclusion, the maximum) 
of the carriers of vertices in $\rho$. 
Given a simplex $\sigma\in\Chr^2\s$, $\Car(\sigma,\s)$ is equal to 
$\Car(\Car(\sigma,\Chr\s),\s)$. $\Car(\sigma,\Chr\s)$
corresponds to the set of all snapshots seen by processes in $\chi(\sigma)$.
Hence, $\Car(\sigma,\s)$ corresponds to the union of all these snapshots. 
Intuitively, it results in the set of all processes \emph{seen} by processes 
in $\chi(\sigma)$ through the two successive immediate snapshots instances.

Additionnal details can be found in Appendix~\ref{app:topprimer}.

\paragraph*{Simplex agreement and affine tasks.}
In the \emph{simplex agreement task}, processes start on vertices of some complex~$K$, 
forming a simplex $\sigma\in K$, and must output vertices of some subdivision of~$K$, 
$\Sub(K)$, so that outputs form a simplex $\rho$ of $\Sub(K)$ respecting carrier 
inclusion, i.e., $\Car(\rho,K)\subseteq\sigma$.
In the simplex agreement tasks considered in the characterization 
of wait-free task computability~\cite{BG97,HS99}, $K$ is the 
standard simplex $\s$ and the subdivision is usually iterations of $\Chr$.

An \emph{affine task} is a generalization of the simplex agreement
task, where $\s$ is fixed as the input complex and where the output complex 
is a pure non-empty sub-complex of some iteration of the standard 
chromatic subdivision, $\Chr^{\ell}\s$. 
Formally, let $L$ be a pure non-empty sub-complex of $\Chr^{\ell}\s$ 
for some~$\ell\in \mathbb{N}$. The affine task associated with $L$ 
is then defined as $(\s,L,\Delta)$, where, for every face $\sigma \subseteq \s$,
$\Delta(\sigma) = L \cap \Chr^{\ell}(\sigma)$.
Note that $L \cap \Chr^{\ell}(\t)$ can be empty, 
in which case the set of participating processes must increase before 
processes may produce outputs.
Note that, since an affine task is characterized by its output complex, 
with a slight abuse of notation, we use~$L$ for both the affine task 
$(\s,L,\Delta)$ and its ouput complex. 

By running~$m$ iterations of this task, we obtain $L^m$, a 
sub-complex of $\Chr^{\ell m}\s$, corresponding to a subset of 
$\IS^{~\ell m}$ runs (each of the $m$ iterations includes $\ell$ $\IS$ rounds). 
The affine model associated with $L$, denoted $L^*$, 
corresponds to the set of infinite runs of the IIS model 
where every prefix restricted to a multiple of $\ell$ $\IS$ rounds 
belongs to the subset of $\IS^{~\ell m}$ runs associated with $L^m$.
Note that the definition of the affine model~$L^*$ is done by satisfying 
a property on all its prefixes. Hence, affine models are, by construction, compact.

\section{Adversaries and agreement functions}
\label{sec:adv}

In this section, we introduce many results from~\cite{KR17} wich will be instrumental 
for our topological characterization. 

\paragraph*{Adversaries.}

It is convenient to model patterns in which process failures can occur 
using the formalism of \emph{adversaries}~\cite{DFGT11}. 
An adversary $\A$ is defined as a set of possible correct process subsets, 
called \emph{live sets}. An infinite run is \emph{$\A$-compliant} if the 
set of processes that are correct in that run belongs to~$\A$. 
An \emph{(adversarial) $\A$-model} $\A$-model is thus defined as the set of
$\A$-compliant runs. 

An adversary is \emph{superset-closed}~\cite{Kuz12} if each 
superset of a live set of~$\A$ is also an element of $\A$, i.e., 
if $\forall S\in \A$, $\forall S'\subseteq \Pi$, $S\subseteq S' \implies S'\in\A$. 
Superset-closed adversaries provide a non-uniform
generalization of the classical \emph{$t$-resilient} adversaries.
An adversary $\A$ is \emph{symmetric} if it only depends on the \emph{sizes}
of live sets, and not on process identifiers: 
$\forall S \in \A$, $\forall S' \subseteq \Pi$, $|S'|=|S|\implies S'\in\A$.
Introduced as symmetric progress conditions in~\cite{Tau10}, 
symmetric adversaries provide a generalization of $t$-resilience 
and $k$-obstruction-freedom.

The \emph{agreement power} of a model, i.e., 
the smallest $k$ such that $k$-set consensus is solvable, 
was determined for adversaries in~\cite{GK10}
in order to characterize the power of adversaries in solving 
\emph{colorless} tasks~\cite{BG93a}. It is formalized as follows:

\begin{definition}{[Agreement power]}
\label{def:scn}
\begin{equation*}
\setcon(\A)= \begin{cases}
0 &\text{if $\A=\emptyset$}\\
\max\limits_{S\in \A}( \min\limits_{a\in S}(
  \setcon(\A|_{S\setminus\{a\}})  + 1 ))&\text{otherwise}
\end{cases}
\end{equation*}

\end{definition}

\noindent With $\A|_P$ the adversary composed of the live sets of $\A$ included~in~$P$.
As previously shown in~\cite{GK11}, for a superset-closed adversary~$\A$, the agreement power of $\A$ 
is equal to $\HSS(\A)$, where $\HSS(\A)$ is 
the size of the minimal hitting set of $\A$, i.e., a set intersecting with each $L\in\A$. 
For a symmetric adversary $\A$, the agreement power formula reduces to  
$\setcon(\A)= |\{k\in\{1,\dots,n\}:\exists S\in\A,|S|=k\}|$. 

\paragraph*{Agreement functions.}
Consider an AS model $M$ and a function $\alpha$ 
mapping subsets of $\Pi$ to integers in~$\{0,\ldots,n\}$. 
We say that $\alpha$ is the \emph{agreement function} of $M$, if
for each $P\in 2^{\Pi}$, $\alpha(P)$ is the agreement powe of 
the model $M|_P$ consisting of runs of $M$ in which no process 
in $\Pi\setminus P$ participates~\cite{KR17}.
Intuitively,~$\alpha(P)$ is the best level of set consensus that
can be solved adaptively in $M$.
By convention, if~$M|_P$ does not contain any run, then $\alpha(P)$ is equal to $0$.

Let $P\subseteq P'\subseteq \Pi$. We can observe that, by construction, 
the agreement function of a model is \emph{monotonic},
i.e., $\alpha(P)\leq \alpha(P')$ and of \emph{bounded growth},
i.e., $\alpha(P') \leq \alpha(P)+|P'\setminus P|$.
%
It was shown in~\cite{KR17} that the agreement function of $\A$ can be defined 
using the $\setcon$ function: $\alpha(P)=\setcon(\A|_P)$. 	

\paragraph*{Fair adversaries.}
Informally, an adversary is \emph{{\fair}}~\cite{KR17} if a subset~$Q$ of
participating processes $P$ cannot achieve a better set consensus than $P$.
For an adversary $\A$, and~$Q\subseteq P\subseteq \Pi$, we define 
$\A|_{P,Q} = \{S\in\A:(S\subseteq P)\wedge(S\cap Q\neq\emptyset)\}$.
When solving a task, only correct processes must  output. Thus, 
for a task restricted to processes in $Q$, no process has 
to output in a run corresponding to a live set~$L\in\A$ with $L\cap Q=\emptyset$. 

\begin{definition}{[Fairness]}\label{def:fair}
An adversary $\A$ is {\fair} if and only if:
\[\forall P \subseteq \Pi, \forall Q\subseteq P, \setcon(\A|_{P,Q})= min(|Q|,\setcon(\A|_P)){}.\]
\end{definition}

\noindent Superset-closed and symmetric adversaries are fair~\cite{KR17}, 
as some others. Unfortunately, not all adversaries are fair.

\paragraph*{The $\alpha$-model.}
Generalizing the $k$-active-resilient model, 
the $\alpha$-model was introduced to capture agreement functions 
ability to characterize the task computability of (some) models.

\begin{definition}{[$\alpha$-model]} 
  The \emph{$\alpha$-model} is the AS model in which: if $P$ is the 
  participating set, then $\alpha(P)\geq 1$ and at most~$\alpha(P)-1$
  processes in~$P$ are faulty.
\end{definition}

\noindent Intuitively, the $\alpha$-model is the weakest model 
with $\alpha$ as its agreement function. This allows us to 
use agreement functions to characterize models which 
are as weak as the corresponding $\alpha$-model.
It turns out that the task computability of a fair adversary is
captured precisely by the corresponding $\alpha$-model, i.e.,
they solve \emph{the same} set of tasks.

\begin{theorem}{\cite{KR17}}
\label{th:adv:task}
For any {\fair} adversary $\A$, a task is solvable in the
$\A$-model if and only if it is solvable in the $\alpha$-model.
\end{theorem} 

\paragraph*{$\alpha$-adaptive set consensus.}
The abstraction of \emph{$\alpha$-adaptive set consensus}~\cite{KR17}
can be accessed with a single $\mathit{propose}(v)$ operation. 
It ensures that (termination) every operation invoked by a correct process
eventually returns, (validity) every returned value is the argument of a
preceding \textit{propose} invocation, and ($\alpha$-agreement) at any
point of the execution, the number of distinct returned values does not exceed~$\alpha(P)$, 
with $P$ the current participating set. 
This abstraction allows us to define yet another family of models, 
equivalent with $\alpha$-model, and hence, with adversarial $\A$-models.

\begin{definition}{\small{[$\alpha$-set consensus model]}}
  The \emph{$\alpha$-set consensus} model is the AS model in which, 
  if $P$ is the participating set then $\alpha(P)\geq 1$, and
  processes have access to $\alpha$-adaptive set-consensus objects.
\end{definition}

\begin{theorem}{\cite{KR17}}
\label{read/writeAndConsensus}
A task is solvable in the $\alpha$-model if and only if
it is solvable in the $\alpha$-set-consensus model.
\end{theorem}

\noindent Hence, for a fair adversary $\A$ and its agreement function~$\alpha$
the $\A$-model, the $\alpha$-set-consensus model and the $\alpha$-model 
can all be used interchangeably for task solvability issues.

\section{Defining the affine task for a fair adversary}
\label{sec:def}

Given a fair adversary $\A$ and its agreement function $\alpha$, 
we define the affine task $\R_\A$, a sub-complex of~$\Chr^2\s$, 
which will be shown to capture the task computability of the $\A$-model.

\paragraph*{Agreement and contention simplices.}

For a vertex~$v\in\Chr^2\s$, let $\mathit{View}^1(v)$ 
and $\mathit{View}^2(v)$  be 
the sets of processes seen by the process $\chi(v)$ in, respectively, the first and
the second $\IS$ (we call these~$\mathit{View}^1$ and $\mathit{View}^2$). Formally, 
$\mathit{View}^2(v)= \Car(v,\Chr\s)$ and $\mathit{View}^1(v)=\Car(v',\s)$ 
with~$v'\in\Car(v,\Chr\s)$ such that $\chi(v)=\chi(v')$.

The idea behind the definition of these prohibited simplices is simple.
In an execution, processes can only decide on the proposal they observed. 
Therefore, in an execution, if a process~$p$ sees only itself, 
other processes should return~$p$'s proposal to hope reaching an agreement with $p$.
In~$\Chr^2\s$, if~$p$ is executed alone, then it has the smallest $\mathit{View}^1$
and $\mathit{View}^2$. Thus all processes would observe~$p$'s $\mathit{View}^1$.
Therefore, a natural way to try to reach an agreement among processes 
is to adopt the proposal from the process observed with the smallest $\mathit{View}^1$.
Moreover, as processes may share the same view, it is even better 
to deterministically select a value from the smallest $\mathit{View}^1$ itself. 

We formalize the intuitive description of contention simplices as follows:
In a simplex $\delta\in\Chr^2\s$, we say that vertices $v$ and $v'$ are \emph{contending}
if their $\mathit{View}^1$ and $\mathit{View}^2$ are ordered in the
opposite way:  $\mathit{View}^1(v)$ is a proper subset of
$\mathit{View}^1(v')$ and  $\mathit{View}^2(v')$ is a proper subset 
of~$\mathit{View}^2(v)$, or vice versa.
If every two vertices of $\delta$ are contending, then we say that
$\delta$ is a \emph{$2$-contention} simplex. Let $\mathit{Cont}_2$ 
be the set of $2$-contention simplices, formally:

\begin{definition}{[$\mathit{Cont_2}$]} 
$\sigma\in \Chr^2 \s: \forall v,v' \in \sigma, v\neq v':$
\[((\mathit{View}^1(v) \subsetneq \mathit{View}^1(v'))\wedge(\mathit{View}^2(v') \subsetneq \mathit{View}^2(v)))
\vee
\]\[
((\mathit{View}^1(v') \subsetneq \mathit{View}^1(v))\wedge(\mathit{View}^2(v) \subsetneq \mathit{View}^2(v'))){}.\]
\end{definition}

\noindent $\mathit{Cont}_2$ is inclusion-closed: any face 
of a $2$-contention simplex is also in $\mathit{Cont}_2$. 
Thus, $\mathit{Cont}_2$  is a complex:  the
\emph{$2$-contention complex} (depicted for a $3$-processes system in Figure~\ref{fig:Contention}).
Particular executions of two~IS rounds are also represented in Figures~\ref{fig:ContEx1} and~\ref{fig:ContEx2}. 
In these executions, one can see that a couple of processes is contending if the 
execution ``order'' is strictly reversed in the two~IS runs. 

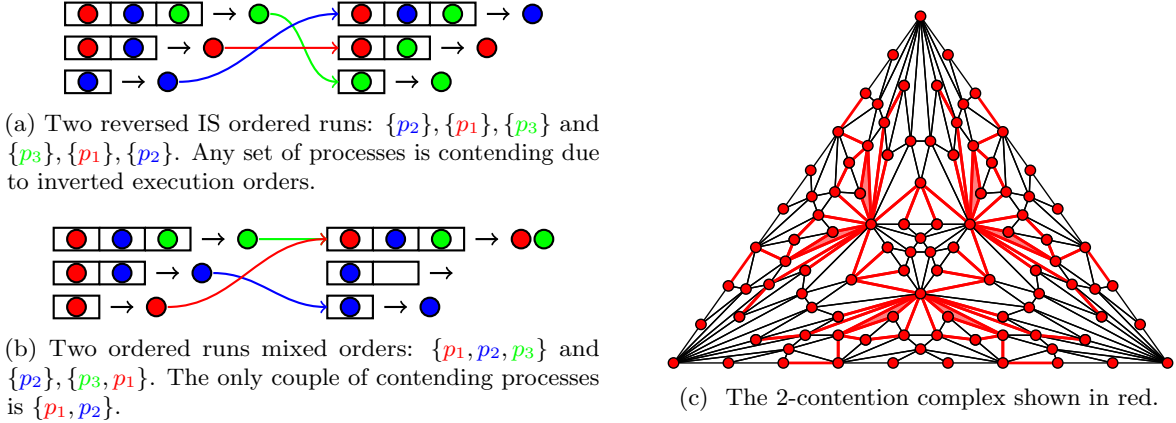
\begin{figure}
\begin{minipage}{0.49\linewidth}
\begin{center}
    \begin{tikzpicture}[scale=0.3]

	\draw[thick] (0,5) -- (2,5) -- (2,4) -- (0,4) -- cycle;
	\draw (1,4.5) node[circle,inner sep=0pt,minimum size=7, fill=red,draw,thick] {};
	\draw[thick] (2,5) -- (4,5) -- (4,4) -- (2,4);
	\draw (3,4.5) node[circle,inner sep=0pt,minimum size=7, fill=blue,draw,thick] {};
	\draw[thick] (4,5) -- (6,5) -- (6,4) -- (4,4);
	\draw (5,4.5) node[circle,inner sep=0pt,minimum size=7, fill=green,draw,thick] {};
	\draw[->, thick] (6.5,4.5) -- (7.5,4.5);
	\draw (8.5,4.5) node[circle,inner sep=0pt,minimum size=7, fill=green,draw,thick] {};
	
	\draw[thick] (0,3.5) -- (2,3.5) -- (2,2.5) -- (0,2.5) -- cycle;
	\draw (1,3) node[circle,inner sep=0pt,minimum size=7, fill=red,draw,thick] {};
	\draw[thick] (2,3.5) -- (4,3.5) -- (4,2.5) -- (2,2.5);
	\draw (3,3) node[circle,inner sep=0pt,minimum size=7, fill=blue,draw,thick] {};
	\draw[->, thick] (4.5,3) -- (5.5,3);
	\draw (6.5,3) node[circle,inner sep=0pt,minimum size=7, fill=red,draw,thick] {};
	
	\draw[thick] (0,2) -- (2,2) -- (2,1) -- (0,1) -- cycle;
	\draw (1,1.5) node[circle,inner sep=0pt,minimum size=7, fill=blue,draw,thick] {};
	\draw[->, thick] (2.5,1.5) -- (3.5,1.5);
	\draw (4.5,1.5) node[circle,inner sep=0pt,minimum size=7, fill=blue,draw,thick] {} ;
    
   \draw[thick] (12,5) -- (14,5) -- (14,4) -- (12,4) -- cycle;
	\draw (13,4.5) node[circle,inner sep=0pt,minimum size=7, fill=red,draw,thick] {};
	\draw[thick] (14,5) -- (16,5) -- (16,4) -- (14,4);
	\draw (15,4.5) node[circle,inner sep=0pt,minimum size=7, fill=blue,draw,thick] {};
	\draw[thick] (16,5) -- (18,5) -- (18,4) -- (16,4);
	\draw (17,4.5) node[circle,inner sep=0pt,minimum size=7, fill=green,draw,thick] {};
	\draw[->, thick] (18.5,4.5) -- (19.5,4.5);
	\draw (20.5,4.5) node[circle,inner sep=0pt,minimum size=7, fill=blue,draw,thick] {} ;

	\draw[thick] (12,3.5) -- (14,3.5) -- (14,2.5) -- (12,2.5) -- cycle;
	\draw (13,3) node[circle,inner sep=0pt,minimum size=7, fill=red,draw,thick] {};
	\draw[thick] (14,3.5) -- (16,3.5) -- (16,2.5) -- (14,2.5);
	\draw (15,3) node[circle,inner sep=0pt,minimum size=7, fill=green,draw,thick] {};
	\draw[->, thick] (16.5,3) -- (17.5,3);
	\draw (18.5,3) node[circle,inner sep=0pt,minimum size=7, fill=red,draw,thick] {};
	
	\draw[thick] (12,2) -- (14,2) -- (14,1) -- (12,1) -- cycle;
	\draw (13,1.5) node[circle,inner sep=0pt,minimum size=7, fill=green,draw,thick] {};
	\draw[->, thick] (14.5,1.5) -- (15.5,1.5);
	\draw (16.5,1.5) node[circle,inner sep=0pt,minimum size=7, fill=green,draw,thick] {};
	
	\draw[green,->, thick] (9,4.5) to[out=0,in=180] (12,1.5);
	\draw[red,->, thick] (7,3) to[out=0,in=180] (12,3);
	\draw[blue,->, thick] (5,1.5) to[out=0,in=180] (12,4.5);
	
	\end{tikzpicture}
	\subcaption{\begin{small}Two reversed IS ordered runs: 
     $\{\color{blue}p_2\color{black}\},\{\color{red}p_1\color{black}\},\{\color{green}p_3\color{black}\}$ and 
     $\{\color{green}p_3\color{black}\},\{\color{red}p_1\color{black}\},\{\color{blue}p_2\color{black}\}$. 
     Any set of processes is contending due to inverted execution orders.\end{small}\label{fig:ContEx1}}

\vspace{1em}
	\begin{tikzpicture}[scale=0.3]
	
	\draw[thick] (0,5) -- (2,5) -- (2,4) -- (0,4) -- cycle;
	\draw (1,4.5) node[circle,inner sep=0pt,minimum size=7, fill=red,draw,thick] {};
	\draw[thick] (2,5) -- (4,5) -- (4,4) -- (2,4);
	\draw (3,4.5) node[circle,inner sep=0pt,minimum size=7, fill=blue,draw,thick] {};
	\draw[thick] (4,5) -- (6,5) -- (6,4) -- (4,4);
	\draw (5,4.5) node[circle,inner sep=0pt,minimum size=7, fill=green,draw,thick] {};
	\draw[->, thick] (6.5,4.5) -- (7.5,4.5);
	\draw (8.5,4.5) node[circle,inner sep=0pt,minimum size=7, fill=green,draw,thick] {};
	
	\draw[thick] (0,3.5) -- (2,3.5) -- (2,2.5) -- (0,2.5) -- cycle;
	\draw (1,3) node[circle,inner sep=0pt,minimum size=7, fill=red,draw,thick] {};
	\draw[thick] (2,3.5) -- (4,3.5) -- (4,2.5) -- (2,2.5);
	\draw (3,3) node[circle,inner sep=0pt,minimum size=7, fill=blue,draw,thick] {};
	\draw[->, thick] (4.5,3) -- (5.5,3);
	\draw (6.5,3) node[circle,inner sep=0pt,minimum size=7, fill=blue,draw,thick] {};
	
	\draw[thick] (0,2) -- (2,2) -- (2,1) -- (0,1) -- cycle;
	\draw (1,1.5) node[circle,inner sep=0pt,minimum size=7, fill=red,draw,thick] {};
	\draw[->, thick] (2.5,1.5) -- (3.5,1.5);
	\draw (4.5,1.5) node[circle,inner sep=0pt,minimum size=7, fill=red,draw,thick] {} ;
    
   \draw[thick] (12,5) -- (14,5) -- (14,4) -- (12,4) -- cycle;
	\draw (13,4.5) node[circle,inner sep=0pt,minimum size=7, fill=red,draw,thick] {};
	\draw[thick] (14,5) -- (16,5) -- (16,4) -- (14,4);
	\draw (15,4.5) node[circle,inner sep=0pt,minimum size=7, fill=blue,draw,thick] {};
	\draw[thick] (16,5) -- (18,5) -- (18,4) -- (16,4);
	\draw (17,4.5) node[circle,inner sep=0pt,minimum size=7, fill=green,draw,thick] {};
	\draw[->, thick] (18.5,4.5) -- (19.5,4.5);
	\draw (20.5,4.5) node[circle,inner sep=0pt,minimum size=7, fill=red,draw,thick] {} ;
	\draw (21.5,4.5) node[circle,inner sep=0pt,minimum size=7, fill=green,draw,thick] {} ;

	\draw[thick] (12,3.5) -- (14,3.5) -- (14,2.5) -- (12,2.5) -- cycle;
	\draw (13,3) node[circle,inner sep=0pt,minimum size=7, fill=blue,draw,thick] {};
	\draw[thick] (14,3.5) -- (16,3.5) -- (16,2.5) -- (14,2.5);
	\draw[->, thick] (16.5,3) -- (17.5,3);
	
	\draw[thick] (12,2) -- (14,2) -- (14,1) -- (12,1) -- cycle;
	\draw (13,1.5) node[circle,inner sep=0pt,minimum size=7, fill=blue,draw,thick] {};
	\draw[->, thick] (14.5,1.5) -- (15.5,1.5);
	\draw (16.5,1.5) node[circle,inner sep=0pt,minimum size=7, fill=blue,draw,thick] {};
	
	\draw[green,->, thick] (9,4.5) to[out=0,in=180] (12,4.5);
	\draw[blue,->, thick] (7,3) to[out=0,in=180] (12,1.5);
	\draw[red,->, thick] (5,1.5) to[out=0,in=180] (12,4.5);
	\end{tikzpicture}
	\subcaption{\begin{small}Two ordered runs mixed orders: 
     $\{\color{red}p_1\color{black},\color{blue}p_2\color{black},\color{green}p_3\color{black}\}$ and 
     $\{\color{blue}p_2\color{black}\},\{\color{green}p_3\color{black},\color{red}p_1\color{black}\}$.
     The only couple of contending processes is $\{\color{red}p_1\color{black},\color{blue}p_2\color{black}\}$.\end{small}\label{fig:ContEx2}}
\end{center}
\end{minipage}
\hfill
\begin{minipage}{0.49\linewidth}
\begin{center}
\includegraphics[scale=1.3]{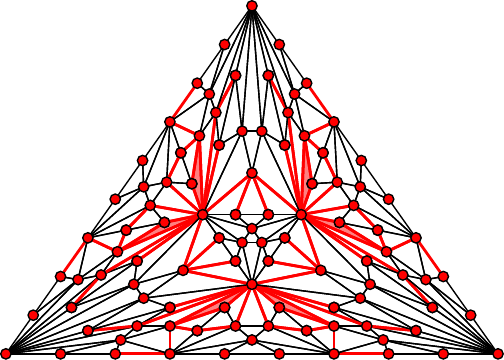}
\subcaption{\label{fig:Contention} The $2$-contention complex shown in red.}
\end{center}
\end{minipage}
\caption{Representation, in a $3$-processes system, of all $2$-contention simplices in $\Chr^2\s$ and some detailed IS runs.\label{fig:ContentionAndExamples}}
\end{figure}

We first show how to restrict $\Chr^2\s$ to obtain an affine task $\R_{k-\mathit{OF}}$,
solvable in the $k$-obstruction-free model, and which allows, in $\R_{k-\mathit{OF}}^*$,
any set of processes to solve $k$-set consensus among themselves. 
As in~\cite{GHKR16}, the idea consists in specifying prohibited 
simplices and take their pure complement as the affine task. 

Intuitively, a contention simplex of size $k$ is one in which, 
in the corresponding run, all of the $k$ processes have distinct 
$\mathit{View}^1$ and each one believes it had the smallest one among them.
Thus, an execution for which all processes would return distinct proposals.
Hence, $\R_{k-\mathit{OF}}$ is defined by prohibiting 
too large contending simplices:

\begin{definition}{[Affine task $\R_{k-\mathit{OF}}$]}
\label{def:Rk}
\[
\R_{k-\mathit{OF}}= \Pc(\{\sigma\in{\mathit{Cont}_2} | {\mathit{dim}}(\sigma)\geq k\},\Chr^2\s){}.
\]
\end{definition}

\noindent See Figure~\ref{fig:1-OF} for $\R_{1-\mathit{OF}}$ in a $3$-process system. 
To see that $\R_{k-\mathit{OF}}$ indeed captures the $k$-obstruction-free adversary, 
one can check, which is not obvious, that the latter definition of $\R_\A$ reduces 
to~$\R_{k-\mathit{OF}}$ when~$\A$ is the $k$-obstruction-free adversary, or, 
alternatively, rely on the proofs from~\cite{GHKR16}.

\paragraph*{Agreement vs. participation.}
Solving a desired level of agreement is no longer sufficient. 
The agreement function of an adversary may define  
different levels of agreement for different participating sets.
In iterated affine tasks, participation is captured by views of the
processes: $\Car(v,\s)$ is the participating set witnessed by process $\chi(v)$.

The naive approach would consists in varying the restriction on the 
size of contention simplices according to the carrier size. 
Such a restriction would indeed provide an affine task 
which is strong enough to solve the desired level of agreement, but 
it would be impossible to solve. Indeed, contention assumes that 
processes with the smallest $\mathit{View}^1$ go first. But 
when the agreement power is equal to $0$, processes must be 
ensured to obtain larger views and hence to let processes with 
larger $\mathit{View}^1$ go first. But letting processes with 
large $\mathit{View}^1$ go first inherently creates contention. 

The idea of the solution consists in switching between resilience 
and concurrency requirements. Indeed, as long as the agreement power 
is steady over the participation, we rely on restrictions made by limiting contention. 
But when the agreement power increases due to an increase of participation, 
we identify a ``witness'' of this new agreement power and require it to 
go first and be seen by other processes. This corresponds to changing the 
selection of the smallest $\mathit{View}^1$ by looking first on $\mathit{View}^1$
``witnessing'' a new agreement level and otherwise, by default, selecting the 
smallest $\mathit{View}^1$. These ``witnesses'' of participation is what we call 
\emph{critical simplices}.

\paragraph*{Critical simplices.}

\begin{figure}
\captionsetup[subfigure]{justification=centering}
  \begin{minipage}[b]{.49\linewidth}
    \centering
\includegraphics[scale=2.33]{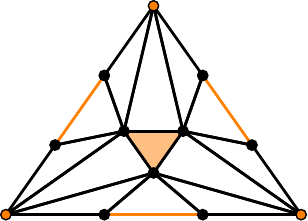}
   \subcaption{\begin{small}Critical simplices for the $\alpha$-model with $\alpha(P)=min(|P|,1)$ ($1$-obstruction-freedom)\end{small}}
  \end{minipage}
  \hfill
  \begin{minipage}[b]{.49\linewidth}
    \begin{center}
\includegraphics[scale=2.33]{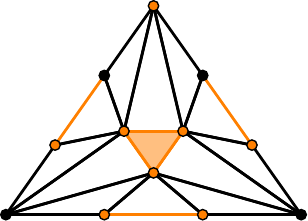}
    \subcaption{\begin{small}Critical simplices for the adversary 
    defined as $\{p_2\}$, $\{p_1,p_3\}$ plus all supersets.\end{small}}
    \end{center}
  \end{minipage}
  \caption{Critical simplices are displayed in orange (with $p_2$ the top vertex, $p_1$ the bottom left vertex and $p_3$ the bottom right vertex).}
  \label{fig:critical_examples}
\end{figure}

The goal here is to identify for each increase of participation a 
new $\mathit{View}^1$ witnessing~it. An easy requirement is that 
this $\mathit{View}^1$ should correspond to a participation level 
associated with the new level of agreement power. But two issues 
must be solved: (1) the provided $\mathit{View}^1$ may be irregular 
and there could be none for a given agreement power; and (2) 
distinct $\mathit{View}^1$ may share the same level of agreement 
power and the smallest one may be different depending on the 
executions.

The idea is to select $\mathit{View}^1$ which are minimal in the 
given execution for some level of agreement power. To do so, the value 
of $\mathit{View}^1$ is not sufficient on its own. But if we know 
that multiple processes all share the same $\mathit{View}^1$, we 
can deduce that all other processes with a strictly smaller view 
must have a $\mathit{View}^1$ corresponding to a lower level of 
agreement power. This solves the second issue, but indirectly 
also the first one. Indeed, if no $\mathit{View}^1$ exists 
for an agreement level, it implies that the smallest view 
for the next level is provided to sufficiently many processes 
to be able to deduce that no process with a smaller $\mathit{View}^1$ 
may obtain a $\mathit{View}^1$ corresponding to the ``missing'' 
level, hence this $\mathit{View}^1$ is a witness of both 
agreement levels.

A \emph{critical set} or \emph{critical simplex} is set of 
processes sharing the same $\mathit{View}^1$ which is 
sufficiently large to ensure that their $\mathit{View}^1$ 
is the smallest one for some level of agreement power.
Formally, a simplex $\sigma \in \Chr \s$ is a \emph{critical simplex} if:
(1)~all its vertices share the same carrier;
and (2)~the set consensus power associated to $\Car(\sigma,\s)$ is 
strictly greater than the set consensus power of $\chi(\Car(\sigma,\s))\setminus \chi(\sigma)$.

\begin{definition}{} $\forall \sigma\in \Chr \s, \mathit{Critical}_\alpha(\sigma) \equiv$
\[ (\forall v \in \sigma: \Car(v,\s)=\Car(\sigma,\s))\wedge
\left(\alpha(\chi(\Car(\sigma,\s))\setminus \chi(\sigma))<\alpha(\chi(\Car(\sigma,\s)))\right){}.
\]
\end{definition}

Examples of critical simplices for two $3$-processes fair models are depicted in Figure~\ref{fig:critical_examples}. The critical simplices are displayed in orange. As it 
can be observed, the set of critical simplices is not inclusion-closed, hence it 
does not define a simplicial complex.

Given a simplex $\sigma\in\Chr\s$, we denote as $\mathcal{CS}_\alpha(\sigma)$ 
the set of critical simplices in $\sigma$, that is~$\mathcal{CS}_\alpha(\sigma)=\{\sigma'\subseteq\sigma:\mathit{Critical}_\alpha(\sigma')\}$. 
Moreover, identifying the set of processes which compose some critical simplex will be 
useful. Thus, let~$\mathcal{CSM}_\alpha(\sigma)$ (critical simplices members) be 
the set of vertices of some $\sigma\in\Chr\s$ which belongs to some critical simplex in $\sigma$, 
formally $\mathcal{CSM}_\alpha(\sigma)=\{\sigma'\in\Cl(\mathcal{CS}_\alpha(\sigma)) : \dim(\sigma')=0\}$.
Note that critical simplices members can be seen also as a sub-complex of $\Chr\s$. Intuitively, 
processes with the smallest $\mathit{View}^2$ should belong to this set.
Similarly we also define the notion of the critical simplex view, $\mathcal{CSV}_\alpha(\sigma)$, 
which corresponds to the set of processes observed by a critical simplex in its~$\mathit{View}^1$. 
It can be simply obtained by taking the carrier in $\s$ of a critical simplex, that is~$\mathcal{CSV}_\alpha(\sigma)= \Car(\mathcal{CSM}_\alpha(\sigma),\s)$.

\paragraph*{Concurrency level.}

\begin{figure}
\captionsetup[subfigure]{justification=centering}
  \begin{minipage}[b]{.49\linewidth}
    \centering
\includegraphics[scale=2.33]{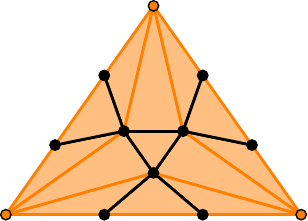}
   \subcaption{\begin{small}For the $\alpha$-model with $\alpha(P)=min(|P|,1)$ ($1$-obstruction-freedom)\end{small}}
  \end{minipage}
  \hfill
  \begin{minipage}[b]{.49\linewidth}
    \begin{center}
\includegraphics[scale=2.33]{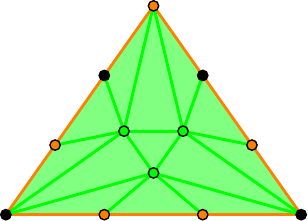}
    \subcaption{\begin{small}For the adversary defined as $\{p_2\}$, $\{p_1,p_3\}$ plus all supersets.\end{small}}
    \end{center}
  \end{minipage}
  \caption{Simplices in black, orange and green are mapped to concurrency levels of $0$, $1$ and~$2$ respectively (with $p_2$ the top vertex, $p_1$ the bottom left vertex and $p_3$ the bottom right vertex).}
  \label{fig:concurrency_map}
\end{figure}

Critical simplices provide a mechanism to select particular $\mathit{View}^1$. 
This can be used to solve agreement protocols with the desired $k$-set consensus 
for an observed participation. But unfortunately this works only when the 
set of processes trying to solve a set-consensus operation was observed 
by the critical simplices, i.e., when processes belong to~$\mathcal{CSV}_\alpha(\sigma)$. 
When this is not the case, processes should be able to solve 
set-consensus operations by themselves. This is where the limition 
on the size of contention simplices will come in. But this 
limitation should still be made according to the observed participation. 
This is done according to the agreement power associated with the observed 
critical simplices. We define this restriction using the following 
notion of \emph{concurrency level}:

\begin{definition}{[Concurrency map]} $\forall\sigma\in\Chr\s:$
\[\mathit{Conc}_\alpha(\sigma) = \max(0\cup\{\alpha(\chi(\Car(\tau,\s))), \tau\in\mathcal{CS}_\alpha(\sigma)\}){}.\]
\end{definition}

Note that we add $0$ to the set of agreement powers in case this set is empty.
The concurrency map is displayed in Figure~\ref{fig:concurrency_map} 
for examples of $3$-processes models. Each simplex of $\Chr\s$ is associated 
with a concurrency level. One can observe that the set of simplices with 
a concurrency level equal to $k$ corresponds to the simplices in the \emph{star}
of the critical simplices associated with an agreement power equal to~$k$ 
and which are not in the \emph{star} of a critical simplex associated with 
a greater agreement power.	
 
\paragraph*{Affine task $\R_\A$.} 
The affine task for a fair adversary $\R_\A\subseteq\Chr^2\s$ is defined as follows: 
\begin{definition}{[$\R_\A$] }\label{def:RA}$\R_\A= \Cl(\{\sigma\in\facets(\Chr^2 \s):\forall \theta \subseteq(\sigma),  P(\theta,\sigma)\}$ with $P$ such that
(with~$\tau= \Car(\theta,\Chr\s)$ and $\rho=\Car(\sigma,\Chr\s)$):
\[ 
P(\theta,\sigma)\equiv\ 
\theta\in\mathit{Cont}_2\wedge(\chi(\theta) \cap 
\chi(\mathcal{CSM}_\alpha(\rho))
\cap\chi(\mathcal{CSV}_\alpha(\tau)))
 = \emptyset
 \implies \mathit{dim}(\theta) < \mathit{Conc_\alpha(\tau)}{}.
\]
\end{definition}

Intuitively, a simplex $\sigma\in\Chr^2\s$ is in $\R_\A$ 
if  and only if any of its ``non-critical'' subsets that cannot ``rely''
on the critical simplices in achieving $\alpha$-adaptive set consensus
has a sufficiently low contention level to solve $\alpha$-adaptive set consensus on
its own.

Examples of affine tasks for $3$-processes $\alpha$-models 
are depicted in Figure~\ref{fig:affineTasks}.

\begin{figure}
\captionsetup[subfigure]{justification=centering}
  \begin{minipage}[b]{.49\linewidth}
    \centering
	\includegraphics[scale=0.7]{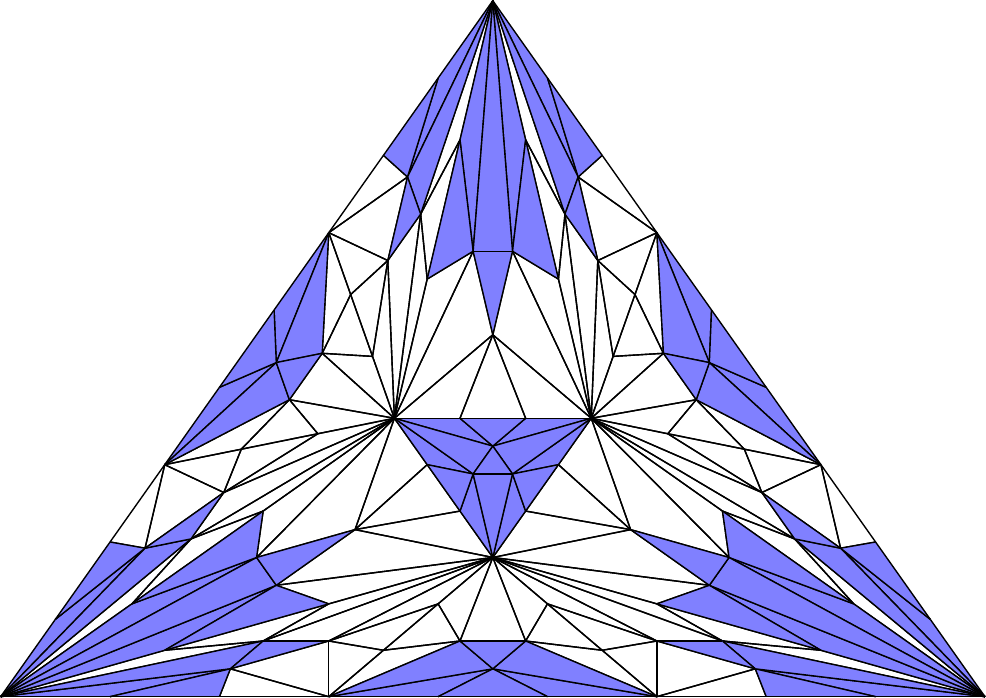}
    \subcaption{\begin{small}Affine task for the $\alpha$-model with $\alpha(P)=min(|P|,1)$. (1-obstruction-freedom) \label{fig:1-OF}\end{small}}
  \end{minipage}
  \hfill
  \begin{minipage}[b]{.49\linewidth}
    \centering
	\includegraphics[scale=0.7]{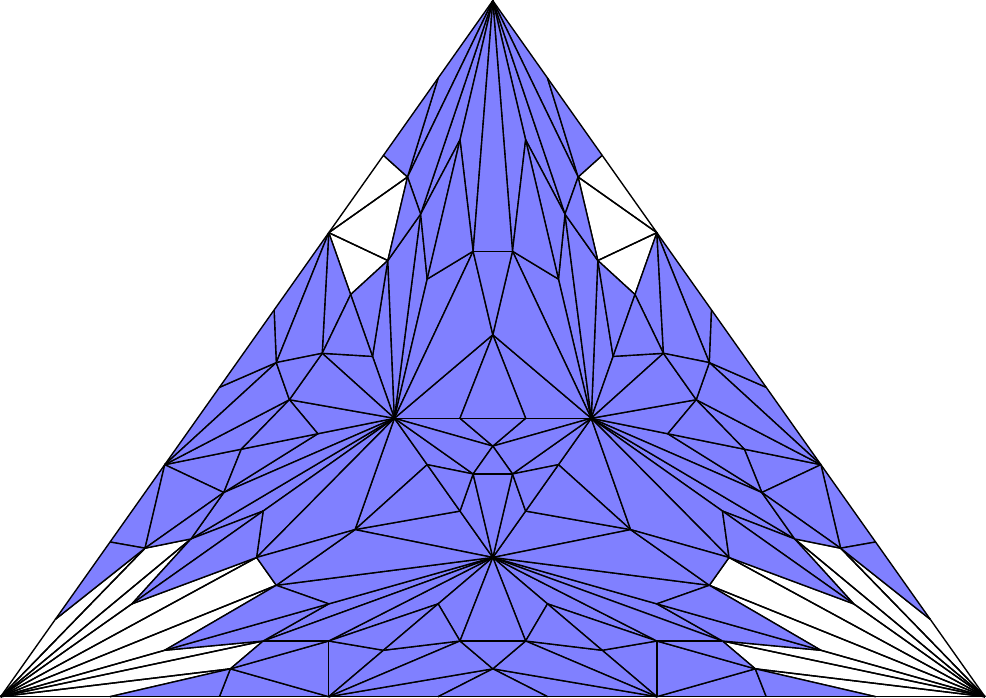}
    \subcaption{\begin{small}Affine task for the adversary 
    defined as $\{p_2\}$, $\{p_1,p_3\}$ plus all supersets.\end{small}}
  \end{minipage}
  \caption{Some examples of affine tasks $\R_\A$ in blue (with $p_2$ the top vertex, $p_1$ the bottom left vertex and $p_3$ the bottom right vertex).}
  \label{fig:affineTasks}
\end{figure}

\section{From the $\alpha$-model to $\R_\A$} 
\label{sec:algo}

To show that any task $T$ solvable in $\R_\A^*$ is solvable in a fair $\A$-model, 
we present an algorithm solving~$\R_\A$ in the $\alpha$-model. 
By iterating this task, we obtain $\R_\A^*$ and can solve~$T$.
	
\subsection{Algorithm Description}
In our solution of $\R_\A$, presented in Algorithm~\ref{Alg:R_A_resolution},
every process accesses two immediate snapshot objects:
$\mathit{FirstIS}$ to which it proposes its initial state, and
$\mathit{SecondIS}$ to which it proposes the outcome of $\mathit{FirstIS}$.
Recall that outcomes of $\mathit{SecondIS}$ form a simplex in $\Chr^2\s$~\cite{Koz12}. 
To ensure that simplices are in~$\R_\A$, after finishing $\mathit{FirstIS}$,
processes wait for their turns to proceed to $\mathit{SecondIS}$.

In this \emph{waiting phase} (Lines~\ref{Alg:RA:FirstIS}--\ref{Alg:RA:WaitConditionEnd}), 
processes check a specific condition on the $\IS$ outcomes that they share with each
others in registers $\ISone[1,\ldots,n]$ and $\IStwo[1,\ldots,n]$.   
Each process $p_i$ periodically checks whether either 
(1) it belongs to a critical simplex by using the formula at Line~\ref{Alg:RA:IsCritical},
or (2)~if the number, computed at Line~\ref{Alg:RA:Rank}, of non-terminated processes 
($\IStwo[j]=\emptyset$) which may have a smaller $\mathit{FirstIS}$ output 
($j\in\ISone[i]$ and $\ISone[j]\neq\ISone[i]$) is smaller that some 
``level of concurrency''. This level of concurrency is computed at 
Line~\ref{Alg:RA:MaxConc} as the maximum between 
(1) the agreement power associated with the $\mathit{View}^1$ 
of the process itself ($\alpha(\ISone[i])$) or 
(2) with the concurrency levels shared using the $\mathit{Conc}$ registers 
by ``terminated'' critical simplices, i.e., a critical simplex with 
all its processes provided with $\mathit{secondIS}$ outputs (Line~\ref{Alg:RA:UpdateConc}).

\begin{algorithm}
\caption{Resolution of $R_\A$ in the $\alpha$-model for process $p_i$.\label{Alg:R_A_resolution}} 
\SetKwRepeat{algo}{$\R_\A(\mathit{input_i})$:}{End $\R_\A$}
\textbf{Immediate Snapshot Objects:} $\mathit{FirstIS}$,
$\mathit{SecondIS}$\;
\textbf{Shared Registers:} $\mathit{Conc}[1],\dots,\mathit{Conc}[n] \in \{0,\dots,n\}$,
$\mathbf{initially}$ $0$\;
$\ISone[1],\dots, \ISone[n] \in 2^{\Pi}$, $\mathbf{initially}$ $\emptyset$ and 
$\IStwo[1],\dots,\IStwo[n] \in 2^{2^{\Pi}}$,
$\mathbf{initially}$ $\emptyset$\;

\vspace{1em}

\algo{}{
	$\ISone[i] \leftarrow \mathit{FirstIS}(\mathbf{input_i})$\;\label{Alg:RA:FirstIS}
	\bf{wait until } $\mathit{crit}\vee(\mathit{rank}<\mathit{conc})$ \bf{with} \label{Alg:RA:WaitCondition}\\
	$\mathit{crit}=(\alpha(\ISone[i])>\alpha(\ISone[i]\setminus\{p_j\in\Pi:\ISone[j]=\ISone[i]\}))$\label{Alg:RA:IsCritical}\\
	\bf{and} $\mathit{rank}=|\{p_j\in\ISone[i]:\IStwo[j]=\emptyset\wedge
          \ISone[j]\neq\ISone[i] \}|$\label{Alg:RA:Rank}\\
    \bf{and} $\mathit{conc}=\max\left(\alpha(\ISone[i]),\max_{j\in \{1,\dots, n\}}(Conc[j])\right)$\label{Alg:RA:MaxConc}\;
	\label{Alg:RA:WaitConditionEnd}
\vspace{1em}
	$\IStwo[i] \leftarrow \mathit{SecondIS}(\ISone[i])$\;	\label{Alg:RA:SecondIS}	
	\If{$\left(\alpha(\ISone[i])>\alpha(\ISone[i]\setminus\left\lbrace p_j\in\Pi:(\ISone[j]=\ISone[i])\wedge(\IStwo[j]\neq\emptyset)\right\rbrace)\right)$}{
		$\mathit{Conc}[i] \leftarrow \alpha(\ISone[i])$\;
	}\label{Alg:RA:UpdateConc}
	$\mathbf{return}(\IStwo[i])$\;	\label{Alg:RA:return}	
}	
\end{algorithm}

Intuitively, the waiting phase is used to ensure that \emph{critical processes},
i.e., members of critical simplices, are prioritized to proceed 
with $\mathit{SecondIS}$ over non-critical ones. 
A process may proceed to its $\mathit{SecondIS}$ as soon as it knows 
that it belongs to some $\mathit{critical}$ simplex ($\mathit{crit}=\mathit{true}$).
A non-critical process is allowed to exit its waiting phase only 
when the number of potentially contending processes is smaller than 
the computed concurrency level ($\mathit{rank}<\mathit{conc}$).
The proof relies mostly on showing that there are enough critical 
simplices to prevent non-critical processes from being blocked in 
the waiting phase.

\subsection{Proof Sketch}

In order to show that Algorithm~\ref{Alg:R_A_resolution} solves $\R_\A$ 
in the $\alpha$-model corresponding to the fair adversary~$\A$, 
we need to show that (1) every correct process eventually outputs 
and that (2) the set of ouputs belongs to a simplex in $\R_\A$. 
Note that as processes execute two consecutive immediate snapshot protocols,
all outputs belong to some simplex in $\Chr^2\s$. Let us consider 
a run of the $\alpha$ model in which the participation is $P$, hence with $\alpha(P)>0$.

To show that outputs belong not only to $\Chr^2\s$ but to $\R_\A$ and 
that all correct processes terminate, we mostly rely on the  
distribution of critical simplices. We are interested in showing that the 
number of processes failures, required to prevent critical simplices from
either appearing in $\ISone$ or completing their $\IStwo$ computation,
scales with the agreement power of the participation. Moreover, 
we want to show that the less processes fail in such a way, 
the higher the maximal agreement power associated with 
a terminated critical simplices.

A process failure may prevent multiple critical simplices to 
terminate. Indeed, a process may be included in multiple critical 
simplices, and thus, its failure would prevent multiple critical 
simplices from terminating. This is why we are interested not in 
the distribution of critical processes or critical simplices, 
but instead, in the minimal hitting set size for 
the set of critical simplices. Let us recall that an hitting set 
of a set of sets $\mathcal{Q}$, is a set intersecting with all 
sets from $\mathcal{Q}$, and that $\HSS$ denotes the minimal hitting set size.
More precisely, we want to know the minimal hitting set size 
of (1) any subset of the participation and (2) of the set of 
critical simplices associated with an agreement 
power greater than or equal to some level $l$, i.e., 
$\{\theta\in\CS(\sigma), \alpha(\chi(\Car(\theta,\s)))\geq l\}$.

\subsection{Distribution of Critical Simplices}

Let us first look at the case in which no participating process
fails before updating its $\ISone$ output to the memory. In this 
case, the set of $\ISone$ views forms a simplex $\sigma\in\Chr\s$ such 
that~$\chi(\sigma)=\chi(\Car(\sigma,\s))$: The observed processes 
include all participating processes (inclusion property) but no others. 
In this setting we can show that the minimal hitting set size of 
the set of critical simplices associated with an agreement power 
greater than or equal to some level $l$, is greater than or equal to
the agreement power of the participation minus $l-1$, i.e., 
$\alpha(\chi(\sigma))-l+1$:

\begin{lemma}{[Distribution of critical simplices]:\label{lem:CSDistribution}}
$\forall\sigma\in\Chr\s,\forall l\in\mathbb{N}:$
\[\chi(\sigma)=\chi(\Car(\sigma,\s)) \implies 
\alpha(\chi(\sigma))-l+1\leq \HSS(\{\theta\in\CS(\sigma), \alpha(\chi(\Car(\theta,\s)))\geq l\}){}.\]
\end{lemma}

\begin{proof} 
Let us fix some integer $l>0$. 
To show Lemma~\ref{lem:CSDistribution}, we proceed by an induction on $\sigma$ 
using the lexicographical order on $(\alpha(\chi(\sigma)),|\chi(\sigma)|)$.
For any simplex~$\sigma$, such that~$\alpha(\chi(\sigma))< l$, 
the result is trivial as for any (possibly empty) set $\mathcal{Q}$, we have 
$\HSS(\mathcal{Q})\geq 0$. Now consider a simplex~$\sigma\in\Chr\s$ such 
that~$\chi(\sigma)=\chi(\Car(\sigma))$ and $\alpha(\chi(\sigma))= k \geq l$. 
Let us assume by induction that for all~$\sigma'\in \Chr \s$, 
if~$(\alpha(\chi(\sigma')),|\chi(\sigma')|)<_{lex}(\alpha(\chi(\sigma')),|\chi(\sigma)|)$, 
then we have:
\[
\chi(\sigma)=\chi(\Car(\sigma',\s)) \implies
\alpha(\chi(\sigma'))-l+1 \leq \HSS(\{\theta\in\CS(\sigma'), \alpha(\chi(\Car(\theta,\s)))\geq l\}){}.
\]
Now consider the face $\tau$ of $\sigma$ consisting of all vertices of $\sigma$ with the
same carrier as $\sigma$, i.e., $\tau=\{v\in\sigma, \Car(v,\s)=\Car(\sigma,\s)\}$. 
Let $\beta$ be the complement of $\tau$, i.e., $\beta=\sigma\setminus\tau$.
Note that~$\tau\neq\emptyset$, due to the containment property, and that, 
$\chi(\Car(\beta,\s))=\chi(\Car(\sigma,\s))\setminus\chi(\tau)$, 
due to the immediacy property. 
Therefore, we obtain that $\chi(\Car(\beta,\s))=\chi(\sigma)\setminus\chi(\tau)$, 
and so that~$\chi(\Car(\beta,\s))=\chi(\beta)$. 
As $(\alpha(\chi(\beta)),|\chi(\beta)|)<_{lex}(\alpha(\chi(\sigma)),|\chi(\sigma)|)$, 
we obtain that:
\begin{equation}
\alpha(\chi(\beta))-l+1 \leq 
\HSS(\{\theta\in\CS(\beta), \alpha(\chi(\Car(\theta,\s)))\geq l\}){}. 
\label{eq:Induct}
\end{equation}
Two cases may arise:
\begin{enumerate}
\item If $\alpha(\chi(\beta))=\alpha(\chi(\sigma))$, then, as $\beta\subseteq\sigma$  
we get that $\CS(\beta)\subseteq\CS(\sigma)$, hence,
we can derive from Equation~\ref{eq:Induct} that:
\[
\alpha(\chi(\sigma))-l+1 \leq \HSS(\{\theta\in\CS(\sigma), \alpha(\chi(\Car(\theta,\s)))\geq l\}){}.
\]
\item If $\alpha(\chi(\beta))<\alpha(\chi(\sigma))$, 
then let $m=\alpha(\chi(\sigma))-\alpha(\chi(\beta))>0$ and 
let us consider any subset~$\tau'$ of $\tau$ such that $|\tau'|>|\tau|-m$.
By construction we have $\Car(\tau',\s)=\Car(\sigma,\s)$ 
and by assumption we have $\chi(\Car(\sigma,\s))=\chi(\sigma)$, 
and thus, we obtain that $\chi(\Car(\tau',\s))=\chi(\sigma)$. 
Let us recall that $\forall v\in\tau:\Car(v,\s)=\Car(\tau,\s)$,
and therefore $\mathit{Critical}_\alpha(\tau')$ if and only 
if~$\alpha(\chi(\sigma)\setminus\chi(\tau'))<\alpha(\chi(\sigma))$. 

Given a fair adversary, for any $Q\subseteq P$, 
we have $\alpha(P)\geq\alpha(P\setminus Q)\geq\alpha(P)-|Q|$.
Note that this property was shown to be true for any fair model 
in~\cite{KR17} (see Section~\ref{sec:adv}).
Note that this implies that $|\chi(\tau)|\geq m$.
By applying the formula for $P=\chi(\sigma)\setminus\chi(\tau')$ 
and for~$Q=\chi(\tau)\setminus\chi(\tau')$ we get that:
\[\alpha(\chi(\sigma)\setminus\chi(\tau'))\geq \alpha(\chi(\sigma)\setminus\chi(\tau))\geq\alpha(\chi(\sigma)\setminus\chi(\tau'))-|\chi(\tau)\setminus\chi(\tau')|{}.
\]
But by construction $\chi(\sigma)\setminus\chi(\tau)=\chi(\beta)$
and $|\chi(\tau)\setminus\chi(\tau')|<m$, thus we obtain that:
\[\alpha(\chi(\sigma)\setminus\chi(\tau))\geq\alpha(\chi(\sigma)\setminus\chi(\tau'))-|\chi(\tau)\setminus\chi(\tau')| \implies \alpha(\chi(\sigma)\setminus\chi(\tau'))<\alpha(\chi(\beta))+m.
\]
As $m=\alpha(\chi(\sigma))-\alpha(\chi(\beta))$, we obtain that $\alpha(\chi(\sigma)\setminus\chi(\tau'))<\alpha(\chi(\sigma))$, and hence, that~$\mathit{Critical}_\alpha(\tau')$.
Since by construction $\beta=\sigma\setminus\tau$, 
we have the following inequality: $\HSS(\CS(\sigma))\geq
\HSS(\CS(\tau))+\HSS(\CS(\beta))$. 
Moreover, as $\alpha(\chi(\sigma))\geq l$, we obtain:
\begin{multline}
\HSS(\{\theta\in\CS(\sigma), \alpha(\chi(\Car(\theta,\s)))\geq l\})\geq 
\\
\HSS(\{\theta\in\CS(\beta), \alpha(\chi(\Car(\theta,\s)))\geq l\})+\HSS(\CS(\tau))
\label{eq:Partition}
\end{multline}
But as any subset $\tau'$ of $\tau$, such that $|\tau'|>|\tau|-m$, is 
critical, we have:
\[
\HSS(\CS(\tau))\geq
\HSS(\{\tau'\subseteq\tau, |\chi(\tau')|> |\chi(\tau)|-m\}){}.
\]
Moreover, since $|\chi(\tau)|\geq m$, 
we have $\HSS(\{\tau'\subseteq\tau, |\chi(\tau')|> |\chi(\tau)|-m\})=m$, 
and hence, that $\HSS(\CS(\tau))\geq m$. 
With $m=\alpha(\chi(\sigma))-\alpha(\chi(\beta))$ and 
Equations~\ref{eq:Induct} and~\ref{eq:Partition}, we obtain:
\begin{eqnarray*}
\alpha(\chi(\sigma))-l+1 &= & (\alpha(\chi(\beta))-l+1)+m\\
&\leq &
\HSS(\{\theta\in\CS(\beta), \alpha(\chi(\Car(\theta,\s)))\geq l\})+\HSS(\CS(\tau))\\
&\leq & 
\HSS(\{\theta\in\CS(\sigma), \alpha(\chi(\Car(\theta,\s)))\geq l\})
\end{eqnarray*}
\end{enumerate}
\end{proof}
 
The result of Lemma~\ref{lem:CSDistribution} can be used to generalize it 
for cases in which not all participating processes shared their \ISone
outputs to the memory. If so, the minimal hitting set size 
decreases proportionally with the number of missing outputs:

\begin{corollary} For any $\sigma\in \Chr\s$, we have:\label{cor:CSDistributionWFailures}
\[\alpha(\chi(\Car(\sigma,\s)))-l-|\chi(\Car(\sigma,\s))\setminus\chi(\sigma)|+1\leq
\]\[
\HSS(\{\theta\in\CS(\sigma), \alpha(\chi(\Car(\theta,\s)))\geq l\}){}.
\]
\end{corollary}
\begin{proof}
Consider some $\sigma\in\Chr\s$. By construction, $\sigma$ is a sub-simplex of 
some simplex $\sigma'$ such that $\chi(\Car(\sigma,\s))=\chi(\Car(\sigma',\s))=\chi(\sigma')$. 
Hence, we can apply Lemma~\ref{lem:CSDistribution} on $\sigma'$ and obtain that:
\begin{equation}
\alpha(\chi(\Car(\sigma,\s)))-l+1\leq \HSS(\{\theta\in\CS(\sigma'), \alpha(\chi(\Car(\theta,\s)))\geq l\}){}.
\label{eq:Completed}
\end{equation}
But $\CS(\sigma)\subseteq\CS(\sigma')$ and thus 
given $H$ a minimal hitting set of $\CS(\sigma')$, 
$H\cup(\chi(\sigma)\setminus\chi(\sigma'))$ is an hitting set of 
$\CS(\sigma')$. Therefore 
$\HSS(\{\theta\in\CS(\sigma), \alpha(\chi(\Car(\theta,\s)))\geq l\})$ 
is greater than or equal to 
$\HSS(\{\theta\in\CS(\sigma'), \alpha(\chi(\Car(\theta,\s)))\geq l\})+|\chi(\sigma)\setminus\chi(\sigma')|$, and thus, is greater than or equal to 
$\HSS(\{\theta\in\CS(\sigma'), \alpha(\chi(\Car(\theta,\s)))\geq l\})+|\chi(\Car(\sigma,\s))\setminus\chi(\sigma)|$. 
Using this in Equation~\ref{eq:Completed} gives us the property of Corollary~\ref{cor:CSDistributionWFailures}.
\end{proof}

\subsection{Algorithm Liveness}

Corollary~\ref{cor:CSDistributionWFailures} is a generalization of 
Lemma~\ref{lem:CSDistribution} to account for a partial set of first immediate 
snapshot outputs. This can be used to show the liveness of the algorithm:

\begin{lemma}\label{lem:Ra_AlgoLiveness}
Algorithm~\ref{Alg:R_A_resolution} provides outputs to all correct processes
in any $\alpha$-model.
\end{lemma}

\begin{proof}
Let $P$ be the participating set and let us assume that there is 
a correct process which never terminates. 
Let $p$ be the correct processes which does not terminate 
with the smallest $\ISone$ view, let $v\in\Chr\s$ 
be the vertex corresponding to its $\ISone$ view, and 
let $\sigma\in\Chr\s$ be the simplex corresponding to the set of 
IS$^1$outputs when $\ISone$ has been updated for the last time.

Due to the immediacy property, processes in $\chi(\Car(v,\s))$ must be associated 
with a vertex~$v'$ such that $\Car(v',\s)\subseteq\Car(v,\s)$, and therefore, 
with $\alpha(\chi(\Car(v',\s)))\leq\alpha(\chi(\Car(v,\s)))$.
Hence, in any completion of $\sigma$ to a simplex~$\sigma'\in\Chr\s$ 
to include the processes which are in $\chi(\Car(v,\s))$ but not in $\chi(\sigma)$, the 
set of critical simplices associated with an agreement power strictly greater than 
$\alpha(\chi(\Car(v,\s)))$ does not change. 
Thus applying Corollary~\ref{cor:CSDistributionWFailures} to any such completion $\sigma'$ of~$\sigma$,
we obtain that, for any $l>\alpha(\chi(\Car(v,\s)))$:
\[
\alpha(\chi(\Car(\sigma,s)))-l-|\chi(\Car(\sigma,\s))\setminus(\chi(\sigma)\cup\chi(\Car(v,\s)))|+1\leq 
\]\[
\HSS(\{\theta\in\CS(\sigma), \alpha(\chi(\Car(\theta,\s)))\geq l\}){}.
\]
Moreover, any process in $P\setminus\chi(\Car(\sigma,s))$ must have failed. 
Thus, in $\chi(\Car(\sigma,s))$ at most $\alpha(P)-1-(|P\setminus\chi(\Car(\sigma,s))|)$ 
processes may fail. Let us recall from the proof of Lemma~\ref{lem:CSDistribution},
that for the agreement function of any fair adversary, and
for any $Q\subseteq P$, we have $\alpha(P)\geq\alpha(P\setminus Q)\geq\alpha(P)-|Q|$.
Thus we can derive, by using $Q=P\setminus\chi(\Car(\sigma,s))$, that
at most $\alpha(\chi(\Car(\sigma,s)))-1$ processes in $\chi(\Car(\sigma,s))$ may fail.

Let $m_1=|\chi(\Car(\sigma,\s))\setminus(\chi(\sigma)\cup\chi(\Car(v,\s)))|$, 
be the number of processes from $\chi(\Car(\sigma,\s))$ which (1) fail before 
updating their $\ISone$ to the memory and (2) are not included in the $\ISone$ view of $p$.
Let $m_2$ be the number of critical processes, associated with an agreement power 
strictly greater than $\alpha(\chi(\Car(v,\s)))$, which fail after updating 
their $\ISone$ but before updating their $\IStwo$.  

Let us now assume that $\alpha(\chi(\Car(\sigma,s)))-\alpha(\chi(\Car(v,s)))> m_1+m_2$, 
then by selecting $l=\alpha(\chi(\Car(\sigma,s)))-m_2-m_1$, we have $l>\alpha(\chi(\Car(v,s)))$), 
and hence, we obtain that:
\[
\HSS(\{\theta\in\CS(\sigma), \alpha(\chi(\Car(\theta,\s)))\geq \alpha(\chi(\Car(\sigma,s)))-m_2-m_1\})\geq m_2+1{}.
\]
If no critical simplex in 
$\{\theta\in\CS(\sigma), \alpha(\chi(\Car(\theta,\s)))\geq 
\alpha(\chi(\Car(\sigma,s)))-m_2-m_1\}$ terminates, 
one process from each of these critical simplices failed after updating 
its $\ISone$ but before updating its $\IStwo$, thus an hitting set failed. 
As only $m_2$ such processes may fail and as an hitting set must be greater than $m_2+1$, 
a critical simplex associated with an agreement power greater than or equal to 
$\alpha(\chi(\Car(\sigma,s)))-m_2-m_1$ terminates its $\IStwo$. 
Therefore eventually some process updates its $\mathit{Conc}$ register 
(on line~\ref{Alg:RA:UpdateConc})
to at least $\alpha(\chi(\Car(\sigma,s)))-m_2-m_1$. 

Now let us look back at $p$. It fails to terminate and thus 
never succeeds to pass the test on line~\ref{Alg:RA:WaitCondition}. 
Therefore we have that the number of processes seen by $p$
which do not terminate and do not have the same $\ISone$ view as $p$ 
are strictly more than the value of 
$\max(\alpha(\ISone[i]),\max_{j\in\{1,\dots,n\}}(\mathit{Conc}[j]))$,
with $\ISone[i]$ equal to $\chi(\Car(v,s))$. 
As~$p$ is the correct process with the smallest $\ISone$ view 
which does not terminate, it implies that there are strictly more 
than $\max(\alpha(\chi(\Car(v,s))),\max_{j\in\{1,\dots,n\}}(\mathit{Conc}[j]))$ 
failed processes with an $\ISone$ view 
strictly smaller than $p$. These failed processes are neither 
accounted in $m_1$ nor in $m_2$. Therefore, as at most 
$\alpha(\chi(\Car(\sigma,s)))-1$ processes in $\chi(\Car(\sigma,s))$ 
may fail, there are at most $\alpha(\chi(\Car(\sigma,s)))-1-m_1-m_2$ 
such processes which may fail. 
Thus $\alpha(\chi(\Car(\sigma,s)))-m_1-m_2-1\geq 
\max(\alpha(\chi(\Car(v,s))),\max_{j\in\{1,\dots,n\}}(\mathit{Conc}[j]))$.

Two cases may arise:
\begin{itemize}
\item If $\alpha(\chi(\Car(\sigma,s)))-\alpha(\chi(\Car(v,s)))> m_1+m_2$, 
then some process sets its $\mathit{Conc}$ register to a value greater than or equal to 
$\alpha(\chi(\Car(\sigma,s)))-m_2-m_1$ --- A contradiction.
\item Otherwise, 
$\alpha(\chi(\Car(\sigma,s)))-m_1-m_2-1\geq \alpha(\chi(\Car(v,s)))$
and so, we obtain a contradiction with the fact that
$\alpha(\chi(\Car(\sigma,s)))-\alpha(\chi(\Car(v,s)))\leq m_1+m_2$.
\end{itemize}
\end{proof}

\subsection{Algorithm Safety}

Showing the safety of Algorithm~\ref{Alg:R_A_resolution} bears some 
similarities with the liveness proof. In particular, it relies on the 
same Lemma~\ref{lem:CSDistribution} on the distribution of critical 
simplices.

\begin{lemma}\label{lem:Ra_AlgoSafety}
The set of outputs provided by Algorithm~\ref{Alg:R_A_resolution} forms 
a valid simplex in $\R_\A$.
\end{lemma}

\begin{proof}
Consider any execution of Algorithm~\ref{Alg:R_A_resolution}. Except for the wait-phase, 
processes execute two rounds of an immediate snapshot protocol. 
Therefore the set of outputs forms a simplex in $\sigma\in\Chr^2\s$. 
Without loss of generality, we can assume that no process fails 
and thus that $\mathit{dim}(\sigma)=n-1$. Indeed, if $\sigma\not\in\R_\A$, then 
if failed processes were just slow and resumed their execution and terminate, 
it would produce~$\sigma'\not\in\R_\A$.
Let us assume by contradiction that $\sigma\not\in\R_\A$, this implies that there exists~$\theta\subseteq\sigma$ such that (for~$\tau=\Car(\theta,\Chr\s)$ and~$\rho=\Car(\sigma,\Chr\s)$):
\[
(\theta\in{\mathit{Cont}_2})\wedge ((\chi(\theta) \cap (\chi(\CSM(\rho))
\cup\chi(\CSV(\tau)))
 = \emptyset)\wedge(
\mathit{dim}(\theta) \geq \mathit{Conc}_\alpha(\tau)){}.
\] 
As $\theta\in{\mathit{Cont}_2}$, we can order the processes associated with 
vertices from $\theta$ according to their IS$^2$ view (or $\Car(v,\Chr\s)$). 
Let~$q_1,\dots,q_k$ be this ordered set of processes. 
As $q_1$ has the smallest IS$^2$ view, and as $\theta\in{\mathit{Cont}_2}$,
$q_1$ also has the largest IS$^1$ view. 

Consider the state of the execution at the time where $q_1$ successfully passes
the test on Line~\ref{Alg:RA:WaitCondition}. To pass this test, $q_1$ witnessed $\ISone$, 
$\mathit{Conc}$ and $\IStwo$ states such that (with $q_1=p_i$):
\[
(\alpha(\mathit{IS1}[i])>\alpha(\mathit{IS1}[i]\setminus\{p_j\in\Pi:\mathit{IS1}[j]=\mathit{IS1}[i]\}))\vee
\]\[
\left(|\{p_j\in\mathit{IS1}[i]: 	  \mathit{IS2}[j]=\emptyset\wedge
          \mathit{IS1}[j]\neq\mathit{IS1}[i] \}|< 
          \mathit{max}\left(\alpha(\mathit{IS1}[i]),\mathit{max}_{j\in \{1,\dots, n\}}(Conc[j])\right)
	\right)
\]
If $(\alpha(\mathit{IS1}[i])>\alpha(\mathit{IS1}[i]\setminus\{p_j\in\Pi:\mathit{IS1}[j]=\mathit{IS1}[i]\}))$, then it implies that $q_1$ belongs to a critical simplex. Indeed, it would belong to a set of processes sharing the same IS$^1$ view and such that, removing this set of processes 
from their IS$^1$ view would result in a set with a strictly smaller agreement power.
But this would contradict $\chi(\theta) \cap \chi(\CSM(\Car(\sigma,\Chr\s)))\neq\emptyset$
as it would include $q_1$. Therefore we have:
\[
|\{p_j\in\mathit{IS1}[i]: 	  \mathit{IS2}[j]=\emptyset\wedge
          \mathit{IS1}[j]\neq\mathit{IS1}[i] \}|< 
          \mathit{max}\left(\alpha(\mathit{IS1}[i]),\mathit{max}_{j\in \{1,\dots, n\}}(Conc[j])\right)
\]
Two cases may arise:
\begin{itemize}
\item $\mathit{max}\left(\alpha(\mathit{IS1}[i]),\mathit{max}_{j\in \{1,\dots, n\}}(Conc[j])\right)\neq\alpha(\ISone[i])$: In this case, 
a register in $\mathit{Conc}$ was set on 
Line~\ref{Alg:RA:UpdateConc} to a value greater than $\alpha(\ISone[i])$. 
It implies that a critical simplex associated with an agreement level 
strictly greater than $|\{p_j\in\mathit{IS1}[i]: \mathit{IS2}[j]=\emptyset\wedge\mathit{IS1}[j]\neq\mathit{IS1}[i] \}|$ terminated its computation and thus is included in $\Car(\theta,\Chr\s)$. But we can observe that $(\chi(\theta)\setminus\{q_1\})\subseteq\{p_j\in\mathit{IS1}[i]: \mathit{IS2}[j]=\emptyset\wedge\mathit{IS1}[j]\neq\mathit{IS1}[i] \}$, and hence, that $\dim(\theta)<\mathit{Conc}_\alpha(\tau)$ --- a contradiction with $\sigma\not\in\R_\A$.

\item $\mathit{max}\left(\alpha(\mathit{IS1}[i]),\mathit{max}_{j\in \{1,\dots, n\}}(Conc[j])\right)=\alpha(\ISone[i])$: Let $c$ be the highest
agreement power associated with a terminated critical simplex 
(with $c=0$ if there is no terminated critical simplex is terminated).
Therefore we have $\mathit{Conc}_\alpha(\tau)\geq c$.
Let $\lambda\in\Chr\s$ be the simplex corresponding to the 
set of IS$^1$ views of processes in $\ISone[i]$ 
which shared their IS$^1$ view at the time $q_1$ passed the test on Line~\ref{Alg:RA:WaitCondition}. 
Consider the simplex $\lambda'\in\Chr\s$ corresponding 
to the completion of $\lambda$ with the vertices corresponding to IS$^1$ view of 
the processes in $\chi(\theta)$ which may be missing from $\lambda$. Note 
that, since $q_1$ has the largest IS$^1$ view among processes from $\chi(\theta)$, 
$\chi(\Car(\lambda',\s))=\chi(\Car(\lambda,\s))=\mathit{IS1}[i]$. Moreover, 
since $\chi(\lambda)= \{p_j\in\ISone[i]:\ISone[j]\neq\bot\}$, we obtain that 
$\chi(\lambda')= \{p_j\in\ISone[i]:\ISone[j]\neq\bot\}\cup\chi(\theta)$.
According to Corollary~\ref{cor:CSDistributionWFailures}
applied to $\lambda'$ with $l=c+1$, we obtain that:
\[
\alpha(\ISone[i])-c-|(\chi(\Car(\lambda',c))\setminus\chi(\lambda'))|\leq
\HSS(\{\phi\in\CS(\lambda'): \alpha(\chi(\phi))\geq c+1\}){}.
\]Note that, since there is no terminated critical simplex with an agreement power 
greater than or equal to $c+1$, it implies that one process 
of each critical simplex identified in $\lambda'$ did not 
terminate its IS$^2$, hence a minimal hitting set. Let $S_c$ be this 
minimal hitting of size equal to $\HSS(\{\phi\in\CS(\lambda'): \alpha(\chi(\phi))\geq c+1\}$. Note that $S_c$ does not include any process in $\chi(\theta)$. 
Indeed, given a critical simplex with the same IS$^1$ view as~$q_i$, adding $q_i$ 
to the critical simplex would produce a critical simplex, but by assumption 
processes in $\chi(\theta)$ do not belong to any critical simplex. 
We also have that $S_c$ does not intersect $S_\emptyset = \{p_j\in\ISone[i]:\ISone[j]=\bot\}$. 
Hence, $|S_c| + |S_\emptyset\setminus\chi(\theta)| + |\chi(\theta)|= |S_c\cup S_\emptyset\cup\chi(\theta)|$. 
Therefore, as $|(\chi(\Car(\lambda',c))\setminus\chi(\lambda'))|=|S_\emptyset\setminus\chi(\theta)|$ 
we obtain that~$\alpha(\ISone[i])-c\leq |S_c\cup S_\emptyset\cup\chi(\theta)|-|\chi(\theta)|$.

Let us now check that $S_c\cup S_\emptyset\cup\chi(\theta)\subseteq \{q_1\}\cup S_T$, with 
$S_T = \{p_j\in\mathit{IS1}[i]: \mathit{IS2}[j]=\emptyset\wedge\mathit{IS1}[j]\neq\mathit{IS1}[i] \}$. All are clearly included in $\mathit{IS1}[i]$
by construction. For processes in $S_\emptyset$, since they have their register in $\mathit{IS1}$ equal to 
$\bot$, it is also the case for their register in $\mathit{IS2}$. For processes in $\chi(\theta)$, 
they have a strictly smaller IS$^1$ view by assumption. For the IS$^2$ view, they will have a strictly larger 
view than $q_1$. But since $q_1$ did not start its second immediate snapshot protocol, processes in 
$\chi(\theta)$ could not have terminated it. For processes in $S_c$, they do not share the same 
IS$^1$ view as any process in $\chi(\theta)$ since they are members of critical simplices, in particular, 
they thus have a distinct IS$^1$ view from $q_1$. 
By assumption, they did not terminate their second immediate 
snapshot protocol, and hence also have their $\IStwo$ register still equal to $\bot$. 

Therefore, we have $S_c\cup S_\emptyset\cup\chi(\theta)\subseteq \{q_1\}\cup S_T$, and hence, 
$|S_c\cup S_\emptyset\cup\chi(\theta)|\leq 1+|S_T|$. But since we also have $|S_T|<\alpha(\ISone[i])$ and $\alpha(\ISone[i])-c\leq |S_c\cup S_\emptyset\cup\chi(\theta)|-|\chi(\theta)|$, we obtain that:
\[
|S_T|< |S_c\cup S_\emptyset\cup\chi(\theta)|-|\chi(\theta)|+c\leq|S_T|+1-|\chi(\theta)|+c {}.
\]

Thus $|\chi(\theta)|\leq c$. But recall that $\mathit{Conc}_\alpha(\tau)\geq c$, and so, 
$|\chi(\theta)|\leq \mathit{Conc}_\alpha(\tau)$ --- a contradiction with $\sigma\not\in\R_\A$.

\end{itemize}
\end{proof}

Using Lemmata~\ref{lem:Ra_AlgoLiveness} and~\ref{lem:Ra_AlgoSafety}, we can 
directly derive the correctness of Algorithm~\ref{Alg:R_A_resolution}:

\begin{theorem}
Algorithm~\ref{Alg:R_A_resolution} solves task $\R_\A$ in 
the $\alpha$-model corresponding to the fair adversary~$\A$.
\label{Thm:RalphaSimulation}
\end{theorem}

As for other solutions of affine task, we can iterate this solution in order 
to simulate a run of $\R_\A^*$. Using this simulation we can therefore solve 
any task which is solvable in $\R_\A^*$:

\begin{theorem}
Any task solvable in $\R_\A^*$ is solvable in the $\A$-model.
\label{Affine:Thm:RAtoA}
\end{theorem}

\section{From $\R_{\A}^*$ to the fair adversarial $\A$-model}
\label{sec:alpha}

In this section, we show that any task solvable in the 
fair adversarial $\A$-model can be solved in $\R_\A^*$.
This reduction is much more intricated than in the other direction. 
Indeed, to show that a model is as strong as an affine task based model,
it only suffices to show that any number of iterations of the affine task 
can be solved. In the general case, it is necessary to show that any task
solvable in the target model can be solved and thus that we can emulate an  
algorithm solving any given task. 

To simplify the simulation complexity, we are going to 
show that we can simulate an execution of a shared memory model 
in which the participation $P$ is such that $\alpha(P)>0$ and 
in which $\alpha$-adaptive set consensus can be solved. 
Using the results from~\cite{KR17} (Theorem~\ref{th:adv:task}), we are able to deduce from it 
that any task solvable in a fair adversarial model can be solved in $\R_\A^*$.

\subsection{Simulation Description.}

The main difficulty of the simulation comes from the 
combination of the failure-freedom and the iterative structure of $\R_\A^*$.
A process obtaining small outputs in all iterations, 
often denominated as a ``fast'' process, may never observe the values shared 
by other processes with larger views, comparatively denominated as a ``slow'' processes. 
But as there are no processes failures, eventually, 
all processes must obtain a task output. 
It requires that fast processes make progress with the simulation 
without waiting for slower processes. Slow processes must thus 
wait for faster processes to terminate their simulation before 
being able to make progress with modifying operations.

This first difficulty is resolved by making processes which obtained a task output
in the simulation to use the special value $\bot$ as input for all
further iterations of $\R_\A$. Slower processes are then aware that
processes using $\bot$ do not interfere anymore and that they 
no longer need to witness their modifications of the simulated system state.

Another difficulty relies in the fact that processes may shift
between making shared memory operations and accessing $\alpha$-adaptive
set consensus abstractions. Moreover, processes may be accessing 
distinct $\alpha$-adaptive set consensus abstractions and may 
access them in different orders. Fortunately, set consensus abstractions
are independant of each others and multiple instances can be simulated
in parallel. But memory operations interact with each others 
and a write operation can be safely terminated only once the write value 
is known to be observed by all other processes. Thus a fast process must ensure 
that slower processes are not able to complete write operations
as long as they did not terminate, even when they do not currently 
have a write operation to perform.

\myparagraph{Atomic-snapshot simulation.} 
To simulate the atomic-snapshot memory, we rely upon  
the algorithm proposed in~\cite{GR10-opodis} that simulates a
lock-free atomic-snapshot algorithm in the \emph{iterated} atomic-snapshot model.
%
%
We run the simulation using the \emph{global views} that the
processes obtain at the end of  $\R_\A^*$ iterations, i.e., $\Car(v,\s)$ for their vertices $v\in\R_\A$. 
Recall that these global views satisfy the properties of atomic
snapshots, but not necessarily the properties of immediate snapshots.

In the simulation, every new update performed by a process is assigned
a monotonically growing \emph{sequence number}.   
A terminated process simply stops incrementing its sequence number,
which allows active (non-terminated) processes to make progress.
Without loss of generality, we assume that in the simulated algorithm,
every active process always has a pending memory operation to perform
(intuitively, if there is nothing to write, the process rewrites its last written value).

\myparagraph{Simulating $\alpha$-adaptive set consensus in $\R_\A^*$.}
The $\alpha$-adaptive set consensus simulation in $\R_\A^*$ 
submits in all iterations input, 
a decision estimate for all known set consensus simulations.
For all pending and newly discovered set consensus simulations
for which processes are involved (i.e., for which they are allowed to participate), 
processes update their decision estimate after each iteration of $\R_\A$. 
Processes adopt a deterministically chosen estimate from, if available, 
an $\ISone$ view associated to a critical simplex, and otherwise, 
from the smallest $\ISone$ view they see. Note that only IS1 views including 
a process which may participate to the agreement are considered.
Most of the complexity of the $\alpha$-adaptive set consensus simulation 
lies in this selection of which $\ISone$ view to adopt from. 
This is described extensively in the next section.

A desicion value is committed only when all processes which are involved 
in the $\alpha$-adaptive set consensus abstraction and which are observed 
in a given iteration of $\R_\A^*$ posses a decision estimate. Once, 
the value is committed, the decision estimate will no longer change
and will eventually be returned as output for the $\alpha$-adaptive
set consensus, but processes need to check that the participation 
in the simulated run is high enough before returning the value.

In order to ensure a high enough participation, processes make sure 
that all processes that they witnessed in preceding iterations of $\R_\A^*$
have completed their first simulated write operation. If not, 
processes simulate this write operation themselves. The content 
of this first write operation simply consists of the process initial 
state. Therefore, any process $p$ may simulate this write operation 
(by using the shared memory simulation)
for any other process $q$ as soon as $p$ knows the initial state of $q$.
Once all processes for which the initial state is know are 
participating in the simulated run, processes can safely 
terminate their $\alpha$-adaptive set consensus with the committed value. 


\subsection{$\alpha$-adaptive leader election in $\R_\A$: the $\mu_Q$ map}
 
Let us consider some $\alpha$-adaptive set consensus and 
let $Q$ be the set of processes which 
(1) may participate in the agreement protocol, 
and, (2) did not terminate yet the main simulation. 
Using the structure of $\R_\A$, we construct a map $\mu_Q$ which 
returns to each vertex $v\in\R_\A$, corresponding to 
a process from $Q$ (i.e., with $\chi(v)\in Q$), 
a leader selected among $Q$ for the given iteration of $\R_\A$. 
The map $\mu_Q$ is constructed in two stages. 
The first stage consists in selecting an $\ISone$ view 
which includes a process from $Q$. Two cases may 
happen depending on whether the process observes in~$\R_\A$ a critical simplex
associated with an $\ISone$ view including a process from $Q$ or not:

If the process observes such a critical simplex 
(i.e., $\chi(\CSV(\Car(v,\Chr\s)))\cap Q\neq \emptyset$), it then simply 
returns the smallest $\ISone$ view
of a critical simplex which includes a process from $Q$, using the map $\delta_Q$:
\[\delta_Q=\chi(\min(\{\Car(\sigma',\s): (\sigma'\in CS_\alpha(\Car(v,\Chr\s)): \chi(\Car(\sigma',\s))\cap Q\neq\emptyset)\}){}.\]
Otherwise (if $\chi(\CSV(\Car(v,\Chr\s)))\cap Q= \emptyset$), the process 
returns the smallest observed $\ISone$ view 
which includes a process from $Q$, using the map $\gamma_Q$:
\[\gamma_Q = \chi(\min(\{\Car(v',\s): 
(v'\in \Car(v,\Chr\s))\wedge(\mathit{dim}(v')=0)\wedge(\Car(v',\s)\cap Q\neq\emptyset)\}){}.\]
The second stage then simply consists in selecting, from the selected
$\ISone$ view, the process from $Q$ associated with the smallest identifier, 
let $\min_Q(V)= \min\{p\in V\cap Q\}$ be this map. 
The map $\mu_Q$ is therefore defined as follows:
\[\mu_Q(v) = \mathbf{\ if\ } (\chi(\CSV(\Car(v,\Chr\s)))\cap Q\neq \emptyset)
\mathbf{\ then\ } \textrm{min}_Q\circ\delta_Q
\mathbf{\ else\ } \textrm{min}_Q\circ\gamma_Q{}. \]
Let us first show that, for any vertex $v\in\R_\A$ corresponding to a process in $Q$, 
the map $\mu_Q$ returns a process from $Q$ observed in $\R_\A$ 
(i.e., a process in $\chi(\Car(v,\s))$) :

\begin{property}{[Validity of $\mu_Q$]\label{prop:muQ_validity}}
$\forall v\in\R_\A, {\mathit{dim}}(v)=0, \chi(v)\in Q : $
\[\mu_Q(v)\in\chi(\Car(v,\s))\wedge\mu_Q(v)\in Q{}.\]
\end{property}

\begin{proof}
Let us fix some vertex $v\in\R_\A$ such that $\chi(v)\in Q$.

Let us assume that $\chi(\CSV(\Car(v,\Chr\s)))\cap Q\neq \emptyset$, 
and hence, $\mu_Q(v)=\textrm{min}_Q \circ \delta_Q(v)$.
Let us recall that, given $\sigma\in\Chr\s$, $\CSV(\sigma)$ is equal to
$\Car(\cup_{\sigma'\in \CS(\sigma)}\sigma',\s)$.
But due to carriers inclusion, the carrier of a simplex is equal 
to the carrier of one of its vertices, and so, of any sub-simplex which includes this vertex. 
Thus, as~$\chi(\CSV(\Car(v,\Chr\s)))\cap Q\neq \emptyset$, we have:
\[
\exists \sigma'\in CS_\alpha(\Car(v,\Chr\s)): \chi(\Car(\sigma',\s))\cap Q\neq\emptyset{}.
\]
This implies that $\delta_Q$ has a valid choice for $v$ and can return the minimal 
one, and so that:
\[\exists \sigma'\in CS_\alpha(\Car(v,\Chr\s)): (\delta_Q(v)=\chi(\Car(\sigma',\s)))
\wedge(\chi(\Car(\sigma',\s))\cap Q\neq\emptyset){}.\]
Since $CS_\alpha(\Car(v,\Chr\s))\subseteq\{\sigma\in\Chr\s;\sigma\subseteq\Car(v,\Chr\s)\}$, 
and as $\mu_Q(v)=\textrm{min}_Q \circ \delta_Q(v)$, we obtain that:
\[
\exists \sigma'\subseteq\Car(v,\Chr\s): (\mu_Q(v)=\textrm{min}_Q \circ \chi(\Car(\sigma',\s)))
\wedge(\mu_Q(v)\in Q){}.
\]
As for any simplex $\sigma\in\Chr^2\s$, we have $\Car(\Car(v,\Chr\s),\s)=\Car(v,\s)$, thus Property~\ref{prop:muQ_validity} is verified if~$\chi(\CSV(\Car(v,\Chr\s)))\cap Q\neq \emptyset$.  

Now let us assume that $\chi(\CSV(\Car(v,\Chr\s)))\cap Q= \emptyset$. 
Due to the self-inclusion property, $\exists v'\in \Car(v,\Chr\s))$
such that $\chi(v')=\chi(v)$. The self-inclusion property again implies 
that $\exists v''\in \Car(v',\s)$ such that $\chi(v'')=\chi(v')=\chi(v)$. 
Hence, as $\chi(v)\in Q$, $\exists v'\in \Car(v,\Chr\s)$ 
such that $\chi(\Car(v',\s))\cap Q \neq\emptyset$. 
Thus $\gamma_Q$ has a valid choice for $v$ and can return the minimal one. 
As before, by the transitivity of carriers inclusion, 
the set returned by $\gamma_Q$, and so the process returned by $\mu_Q$, 
is a subset of $\chi(\Car(v,\s))$ which intersects with $Q$.
\end{proof}

Now that we have checked that $\mu_Q$ is well defined, let us 
show that $\mu_Q$ returns a number of distinct leaders (processes)
limited by the agreement power associated with 
processes views in $\R_\A$:

\begin{property}{[Agreement of $\mu_Q$]\label{prop:muQ_agreement}}
$\forall Q\subseteq\Pi, (\forall \sigma\in\R_\A:\mathit{dim}(\sigma)=n-1),(\forall\theta\subseteq\sigma:\chi(\theta)\subseteq Q):$
\[|\{\mu_Q(v):v\in\theta\}|\leq \alpha(\chi(\Car(\theta,\s))){}.\]
\end{property}

Let us first check the following observation stating that 
for any simplex $\sigma\in\Chr\s$, if two critical simplices in $\sigma$ 
are associated with the same agreement power, then they share the same $\ISone$ view:

\begin{lemma}\label{lem:SingleCriticalCarPerAgreementLevel}
$\forall\sigma\in\Chr\s$, $\forall\theta_1,\theta_2\in\CS(\sigma)$: 
\[\alpha(\chi(\Car(\theta_1,\s)))=\alpha(\chi(\Car(\theta_2,\s)))
\implies \Car(\theta_1,\s)=\Car(\theta_2,\s){}.\]
\end{lemma}

\begin{proof}
Let us consider some simplex $\sigma\in\Chr\s$ and some critical simplices 
$\theta_1,\theta_2\in\CS(\sigma)$ such that $\alpha(\chi(\Car(\theta_1,\s)))=\alpha(\chi(\Car(\theta_2,\s)))$.
The inclusion property implies, w.l.o.g., $\Car(\theta_1,\s)\subseteq\Car(\theta_2,\s)$. 
The immediacy property implies either that $\Car(\theta_1,\s)=\Car(\theta_2,\s)$ 
(and thus Lemma~\ref{lem:SingleCriticalCarPerAgreementLevel} is verified) 
or else that~$\chi(\theta_2)\cap\chi(\Car(\theta_1,\s))=\emptyset$. 

Let us now assume that $\chi(\theta_2)\cap\chi(\Car(\theta_1,\s))=\emptyset$. 
Together with $\Car(\theta_1,\s)\subseteq\Car(\theta_2,\s)$, it implies 
that $\Car(\theta_1,\s)\subseteq\Car(\theta_2,\s)\setminus\theta_2$. 
Since agreement functions are regular 
(i.e., the agreement power can only grow with a participation increase), we obtain that
$\alpha(\chi(\Car(\theta_1,\s))\leq\alpha(\chi(\Car(\theta_2,\s)\setminus\theta_2)$. 
But as $\theta_2$ is a critical simplex $\alpha(\chi(\Car(\theta_2,\s)\setminus\theta_2))<
\alpha(\chi(\Car(\theta_2,\s)))$, and we obtain a contradiction:
\[\alpha(\chi(\Car(\theta_1,\s))\leq\alpha(\chi(\Car(\theta_2,\s)\setminus\theta_2)<
\alpha(\chi(\Car(\theta_2,\s)))=\alpha(\chi(\Car(\theta_1,\s))){}.\]
\end{proof}
Let us now prove Property~\ref{prop:muQ_agreement}:
\begin{proof}
Let $\sigma$ be a maximal simplex of $\R_\A$, i.e., $\mathit{dim}(\sigma)=n-1$, 
and let $\theta\subseteq\sigma$ such that $\chi(\theta)\subseteq Q$.

Note that for both $\gamma_Q$ and $\delta_Q$, processes returns the 
$\ISone$ view of a vertex of $\Car(\theta,\Chr\s)$. 
Assume that $\gamma_Q$ and $\delta_Q$ return, for vertices in $\theta$, 
$k\geq 0$ distinct $\ISone$ views which are not the $\ISone$ views of 
some critical simplex in $\Car(\sigma,\Chr\s)$.  
As $\delta_Q$ only returns $\ISone$ views associated with a critical simplex, 
they have been returned by $\gamma_Q$. Let $\beta$ be the subset of $\theta$ 
including all vertices for which~$\gamma_Q$ returns such $\ISone$ views.
As they are returned by $\gamma_Q$, we have $\CSV(\Car(\beta,\Chr\s))\cap Q = \emptyset$.

Consider any two processes $p_1$ and $p_2$ which obtained two distinct such
$\ISone$ views, $V_1$ and $V_2$ respectively (w.l.o.g., let $V_1\subsetneq V_2$).
As $\gamma_Q$ returns the minimal $\ISone$ view intersecting with $Q$, 
a vertex from $\beta$ sees $V_2$ but not $V_1$, and thus, 
$p_2$ has a smaller $\IStwo$ view than $p_1$. 
Therefore $p_1$ and~$p_2$ satisfy the condition to be part of a contention simplex, 
and so, any $k$ processes carrying these $k$ distinct returned~$\ISone$ views 
form a contention simplex. Let $\tau$ be this contention simplex in~$\sigma$. 

As a vertex in $\beta$ saw all these $k$ distinct $\ISone$ views,
we have $\Car(\tau,\Chr\s)\subseteq\Car(\beta,\Chr\s)$.
But since $\CSV(\Car(\beta,\Chr\s))\cap Q = \emptyset$ is satisfied, 
we obtain that $\CSV(\Car(\tau,\Chr\s))\cap Q = \emptyset$. By assumption, 
these $k$ processes are not critical simplices members 
($\chi(\tau)\cap\CSM(\sigma)=\emptyset$). Therefore, the definition of $\R_\A$ 
implies that we have $\mathit{Conc}_\alpha(\Car(\tau,\Chr\s))\geq k$, 
and hence, we obtain that~$\mathit{Conc}_\alpha(\Car(\beta,\Chr\s))\geq k$.

Having $\mathit{Conc}_\alpha(\Car(\beta,\Chr\s))\geq k$, it implies that 
we have $\exists\sigma_c\in \CS(\Car(\beta,\Chr\s))$ such 
that $\alpha(\chi(\Car(\sigma_c,\s)))\geq k$. 
But, since $\chi(\Car(\sigma_c,\s))\subseteq\CSV(\Car(\beta,\Chr\s))$, 
we obtain that~$\chi(\Car(\sigma_c,\s))\cap Q=\emptyset$. As the inclusion 
property implies that any $\ISone$ view must be strictly larger 
to intersect with $Q$, and as there are at most one $\ISone$ view 
associated with a critical simplex by agreement level 
(Lemma~\ref{lem:SingleCriticalCarPerAgreementLevel}), all $\ISone$ 
views corresponding to some critical simplex in~$\Car(\sigma,\Chr\s)$ 
are associated with an agreement power strictly greater than~$k$.

Let $l\geq 0$ be the number of distinct $\ISone$ views corresponding to 
some critical simplex in $\Car(\sigma,\Chr\s)$ which are returned 
by $\delta_Q$ or $\gamma_Q$ for vertices in $\theta$. 
Lemma~\ref{lem:SingleCriticalCarPerAgreementLevel} implies that 
they must be associated with~$l$ distinct agreement powers.
As they must also be associated with agreement powers strictly 
greater than~$k$, one of the returned $\ISone$ views is 
associated with an agreement power greater than or equal to~$k+l$.
Therefore, we have $\alpha(\chi(\Car(\theta,\s)))\geq k+l$.
As the number of distinct $\ISone$ views returned by $\delta_Q$ 
or~$\gamma_Q$ is equal to $k+l$, and as the deterministic 
selection made by $\textrm{min}_Q$ could only reduce the number 
of distinct returned values, we finally obtain that~$|\{\mu_Q(v):v\in\theta\}|\leq \alpha(\chi(\Car(\theta,\s)))$.
\end{proof}

Last, let us also observe that knowing which processes terminated the 
main simulation is not required to compute $\mu_Q$, i.e., that the knowledge 
of which processes belong to $Q$ among the processes observed
in the current iteration of $\R_\A$ is sufficient: 

\begin{property}{[Robustness of $\mu_Q$]}\label{prop:muQ_robustness}
$\forall v\in\R_\A, {\mathit{dim}}(v)=0, \forall Q\subseteq \Pi:
\mu_Q(v) = \mu_{\Car(v,\s)\cap Q}(v)$.
\end{property}

\begin{proof}
This is a direct corollary of the definition of $\delta_Q$ and $\gamma_Q$, 
that for a given vertex $v\in\R_\A$, to compute $\mu_Q(v)$, the knowledge 
of $Q\cap(\Car(v,\s))$ is sufficient. Indeed, $Q$ is only used 
to compute intersections with either $\CSV(\chi(\Car(v,\Chr\s)))$, 
a subset of $\Car(v,\s)$, or with~$\Car(v',\s)$ for a vertex 
$v'\in\Car(v,\Chr\s)$, also a subset of $\Car(v,\s)$.
\end{proof}

\subsection{Correctness of the simulation}

Let us first show that all simulated operations are safe.
Since the composable shared memory simulation is safe, we only need 
to show that simulated $\alpha$-adaptive set consensus operations 
satisfy the validity property (decision values are proposal values) 
and the $\alpha$-agreement property (if $k$ distinct values 
have been returned, then the current participation $P$ is such that $\alpha(P)\geq k$).

\begin{lemma}\label{lem:Validity}
The shared memory and $\alpha$-adaptive set consensus simulation in $\R_\A^*$ is safe.  
\end{lemma}

\begin{proof}
For $\alpha$-adaptive set consensus operations, Property~\ref{prop:muQ_robustness} 
ensures that $\mu_{\Car(v,\s)\cap Q}(v)$ can be used as if it was $\mu_Q(v)$ and thus 
that processes can indeed use $\mu_Q$ to elect a leader in any iteration of $\R_\A$. 
Moreover, Property~\ref{prop:muQ_validity} ensures that a decision estimate 
is either the process proposal value or is adopted from another process with a 
proposal value and thus that the validity property of $\alpha$-adaptive set consensus
is verified. 

At the earliest iteration $R$ of $\R_\A^*$ at which a process commits a decision estimate 
for an $\alpha$-adaptive set consensus, since a committing process only observed 
processes from $Q$ with decision estimates, all processes in $Q$ adopt a decision estimate.
Moreover, Property~\ref{prop:muQ_agreement} states that among any $k$ processes 
adopting $k$ distinct decision estimates at this iteration $R$, one must have observed 
a set of processes associated with an agreement level greater or equal to $k$. 

Before completing an $\alpha$-adaptive set consensus operation, 
processes make sure that all processes they observed are participating in the simulated run 
(by simulating for them a write operation of their initial states). 
Therefore, at the time a $k^{th}$ distinct value is returned for some 
$\alpha$-adaptive set consensus, the participation in the simulated run is 
associated with an agreement power greater than or equal to $k$, hence, 
the $\alpha$-agreement property is verified.
\end{proof}

As we have shown that the simulation is safe, let us also show that it is live, i.e., 
that it provides outputs to all processes. For this, we only need to show that 
a process obtaining the smallest IS view among non-terminated processes eventually 
completes its $\alpha$-adaptive set consensus operation.

\begin{lemma}\label{lem:Progress}
In the shared memory and $\alpha$-adaptive set consensus simulation in $\R_\A^*$, 
all processes eventually terminate.
\end{lemma}

\begin{proof}
as soon as they observe a process with a decision estimate, 
processes adopt a decision estimate for any $\alpha$-adaptive set 
consensus operation they may participate to. Hence, if a process has a 
pending $\alpha$-adaptive set consensus operation and has the 
smallest IS view among non-terminated processes, all competitors 
will have a decision in the next iteration of the affine task. 
But since a process commits a decision estimate when all non-terminated 
processes which participate in an operation share a decision estimate 
during some iteration, a process obtaining the smallest IS view among 
non-terminated processes eventually commits its operation in the 
following affine task iteration and hence resume its shared-memory 
simulation.
\end{proof}

By Lemmata~\ref{lem:Validity} and~\ref{lem:Progress}, 
the simulation that we provide can be used to solve in $\R_\A^*$ any task solvable 
in a shared memory model with access to $\alpha$-adaptive set consensus. 
But it was shown in Section~\ref{sec:adv} that 
the $\alpha$-adaptive set consensus model, the $\alpha$-model and 
a corresponging fair adversarial model are all equivalent in term 
of task solvability. Hence, since we have shown in Theorem~\ref{Affine:Thm:RAtoA}
that all task solvable in~$\R_\A^*$ are solvable in the $\A$-model, 
we obtain the following equivalence result:

\begin{theorem}\label{ref:main}
A task is solvable in the adversarial $\A$-model if and only if it is solvable in $\R_{\A}^*$.
\end{theorem}
We thus obtain the following generalization of the ACT~\cite{HS99}:

\begin{theorem}{[Fair Asynchronous Computability Theorem [FACT]]}\label{ref:fact}
A task $T=(\I,\O,\Delta)$ is solvable in the fair adversarial $\A$-model
if and only if there exists a natural number $\ell$ and a simplicial 
map~$\phi: \R_{\A}^{\ell}(\I) \rightarrow \O$ carried by $\Delta$.
\end{theorem}

\section{Related work}
\label{sec:related}

Inspired by the model of \emph{dependent failures} proposed by 
Junqueira and Marzullo~\cite{JM07-cores}, Delporte et al.~\cite{DFGT11} 
suggested the notion of adversaries and showed that 
adversaries having the same set consensus power agree 
on the set of colorless tasks they solve.

Herlihy and Shavit~\cite{HS99} proposed a characterization 
of wait-free task computability through the existence of 
a simplicial map from a subdivision of the input complex of 
a task $\I$ to its output complex~$\O$. (The reader is referred 
to~\cite{HKR14} for a thorough discussion of the use of 
combinatorial topology in distributed computability.)
Herlihy and Rajsbaum~\cite{HR12} studied colorless task computability
in the special case of \emph{superset-c7-TRlosed} adversaries.
They show that the protocol complex of a superset-closed adversary 
with \emph{minimal core size} $c$ is $(c-2)$-connected.
This result, obtained via an iterative application of the Nerve lemma, 
gives a combinatorial characterization of superset-closed adversaries.
The characterization only applies to colorless tasks, 
and it does not allow us to express the adversary in an affine way.

Gafni et al.~\cite{GKM14} introduced the notion of an affine task
and characterized task computability in \emph{iterated} adversarial models 
via infinite subdivisions of input complexes, assuming a limited notion of 
solvability that only guarantees outputs to ``fast'' 
processes~\cite{Gaf98,BGK14} (i.e., ``seen'' by every other process infinitely often).
The liveness property defined in this paper for iterated models
guarantees outputs for \emph{every} process, which allowed us to establish 
a task-computability equivalence with conventional non-iterated models. 

Saraph et al.~\cite{SHG16} gave a compact combinatorial characterization of 
$t$-resilient task computability. Note that $\A_{t-\mathit{res}}$ is a superset-closed 
(and thus fair) adversary. Our solution of the affine task $\R_{\A}$ 
in the $\alpha$-model is inspired by the $t$-resilient solution of 
$\R_{t-\mathit{res}}$ in~\cite{SHG16}.   
Gafni et al.~\cite{GHKR16} presented affine tasks for the model of 
$k$-set consensus and, thus, $k$-concurrency and $k$-obstruction-freedom, 
which can be expressed as a symmetric and thus fair adversary.

The notions of agreement functions and a fair adversaries 
were introduced by the first two authors in~\cite{KR17}.
One can determine the agreement function of any given adversary
using the formula suggested earlier for the set consensus power~\cite{GK10}. 
It has been shown in~\cite{KR17} that agreement functions 
encode enough information to characterize the task computability 
of any fair adversary. 

A short version of this paper appeared as a conference 
brief announcement~\cite{KRH17-BA}, and as an extended 
version without formal proofs and shorten explanation~\cite{KRH18}.

\section{Concluding remarks}
\label{sec:disc}
This paper generalizes all existing topological characterizations of distributed 
computing models~\cite{HS99,HR12,GKM14,GHKR16,SHG16}.
It applies to all tasks (not necessarily colorless) and all {\fair}
adversarial models (not necessarily $t$-resilience or $k$-obstruction-freedom).
Just as the wait-free characterization~\cite{HS99} implies that 
the $\IS$ task captures the wait-free model, our characterization 
equates any fair adversary with a (compact) affine task 
embedded in the second degree of the standard chromatic subdivision.  

Interestingly, unlike~\cite{SHG16}, we cannot rely on the
\emph{shellability}~\cite{HKR14} (and, thus, link-connectivity) of the affine task.
Link-connectivity of a simplicial complex $\C$ allows us to work 
in the \emph{point set} of its geometrical embedding $|\C|$ and 
use continuous maps (as opposed to simplicial maps that maintain the simplicial structure).
For example, the existence of a continuous map from $|\R_{\A_{t-\mathit{res}}}|$ to
any $|\R_{\A_{t-\mathit{res}}}^k|$ implies that $\R_{\A_{t-\mathit{res}}}$ indeed captures the
general task computability of $\A_{t-\mathit{res}}$~\cite{SHG16}.
In general, however, the existence of a continuous map onto $\C$ only allows us to
converge on a \emph{single} vertex~\cite{HKR14}.
If $\C$ is not link-connected, converging on one vertex allows
us to compute an output only for a single process, and not more. 
Unfortunately, only very special adversaries, such as $\A_{t-\mathit{res}}$, 
have link-connected counterparts (see, e.g., the affine task corresponding 
to $1$-obstruction-freedom in Figure~\ref{fig:1-OF}). 
Instead of relying on link-connectivity, this paper takes an explicit 
algorithmic way of showing that iterations of $\R_\A$ simulate~$\A$.   
An interesting question is to which extent point-set topology
and continuous maps can be applied in affine characterizations.

Given that some models out of this class cannot be grasped by agreement functions 
(see~\cite{KR17} for examples), going beyond fair adversarial models is an important challenge.
In particular, we should be able to account for models in which \emph{coalitions} of participants 
can achieve better levels of set consensus than the whole set.  
Nailed down, this may allow us to compactly capture all ``natural'' models~\cite{GHKR16}, 
such as, e.g., generic adversarial models or the \emph{set consensus collections} models~\cite{DFGK16} 
for which only special cases of $k$-set consensus~\cite{GHKR16} and 
$k$-test-and-set~\cite{KR18-TR} have been, in this sense, understood so far.

\bibliographystyle{abbrv}
\def\noopsort#1{} \def\No{\kern-.25em\lower.2ex\hbox{\char'27}}
  \def\no#1{\relax} \def\http#1{{\\{\small\tt
  http://www-litp.ibp.fr:80/{$\sim$}#1}}}

\newpage

\appendix 

\section{Simplicial complexes}
\label{app:topprimer}

We recall now several notions from combinatorial topology. For more detailed coverage of the topic please refer
to~\cite{Spanier,HKR14}.

A {\em simplicial complex} is a set $V$, together with a collection $C$ of finite non-empty subsets of $V$ such
that:
\begin{enumerate}
\item For any $v \in V$, the one-element set $\{v\}$ is in $C$;
\item If $\sigma \in C$ and $\sigma' \subseteq \sigma$, then $\sigma' \in C$.
\end{enumerate}

The elements of $V$ are called {\em vertices}, and the elements of $C$ are called a {\em simplices}. We usually
drop $V$ from the notation, and refer to the simplicial complex as $C$.


A subset of a simplex is called a {\em face} of that simplex.

A {\em sub-complex} of $C$ is a subset of $C$ that is also a simplicial complex.

The {\em dimension} of a simplex $\sigma \in C$ is its cardinality minus one. The $k$-skeleton of a complex $C$,
denoted $\Skel^k C$, is the sub-complex formed of all simplices of $C$ of dimension $k$ or less.

A simplicial complex $C$ is called {\em pure} of dimension $n$ if $C$ has no simplices of dimension $> n$, and
every $k$-dimensional simplex of $C$ (for $k < n$) is a face of an $n$-dimensional simplex of $C$.


Let $A$ and $B$ be simplicial complexes. A map $f: A \to B$ is called {\em simplicial} if it is induced by a map
on vertices; that is, $f$ maps vertices to vertices, and for any $\sigma \in A$, we have $$ f(\sigma) =
\bigcup_{v \in \sigma} f(\{v\}).$$ A simplicial map $f$ is called {\em non-collapsing} (or {\em
dimension-preserving}) if $\dim f(\sigma) = \dim \sigma$ for all
$\sigma \in A$.

A map $\Phi: A \rightarrow 2^B$ (mapping simplices of $A$ to
sub-complexes of $B$) is called \emph{carrier} if for all $\tau,\sigma
\in A$, we have
$\Phi(\tau\cap\sigma)\subseteq\Phi(\tau)\cap\Phi(\sigma)$.
A simplicial map $\phi: A\rightarrow B$ is said to be \emph{carried by
  a carrier map $\Phi:A\to2^B$} if for all $\sigma\in A$,
$\phi(\sigma) \subset \Phi(\sigma)$.

Any simplicial complex $C$ has an associated {\em geometric realization} $|C|$, defined as follows. Let $V$ be
the set of vertices in $C$. As a set, we let $C$ be the subset of $[0,1]^V = \{\alpha : V \to [0,1]\}$
consisting of all functions $\alpha$ such that $\{ v \in V \mid \alpha(v) > 0 \} \in C$ and $\sum_{v \in V}
\alpha(v) = 1$.
For each $\sigma \in C$, we set $|\sigma| = \{ \alpha \in |C| \mid \alpha(v) \neq 0 \Rightarrow v \in \sigma
\}.$ Each $|\sigma|$ is in one-to-one correspondence to a subset of $\R^n$ of the form $\{(x_1, \dots, x_n) \in
[0,1]^n \mid \sum x_i = 1\}.$ We
put a metric on $|C|$ by $d(\alpha, \beta) = \sum_{v \in V} |\alpha(v) - \beta(v)|.$ 


A non-empty complex $C$ is called {\em $k$-connected} if, for each $m\leq k$, any continuous map of the
$m$-sphere into $|C|$ can be extended to a continuous map over the $(m+1)$-disk.

A {\em subdivision} of a simplicial complex $C$ is a simplicial complex $C'$ such that:
\begin{enumerate}
\item The vertices of $C'$ are points of $|C|$.

\item For any $\sigma' \in C'$, there exists $\sigma \in C$ such that $\sigma' \subset |\sigma|$.

\item The piecewise linear map $|C'| \to |C|$ mapping each vertex of $C'$ to the corresponding point of $C$ is a
homeomorphism.
\end{enumerate}



\myparagraph{Chromatic complexes.}
We now turn to the chromatic complexes used in distributed computing, and recall some notions from \cite{HS99}.

Fix $n \geq 0$. The {\em standard $n$-simplex} $\s$ has $n+1$ vertices, in one-to-one correspondence with $n+1$
{\em colors} $0, 1, \dots, n$. A face $\t$ of $\s$ is specified by a collection of vertices from $\{0, \dots,
n\}$. We view $\s$ as a complex, with its simplices being all possible faces $\t$.

A {\em chromatic complex} is a simplicial complex $C$ together with a non-collapsing simplicial map $\chi: C \to
\s$. Note that $C$ can have dimension at most $n$. We usually drop $\chi$ from the notation. We write $\chi(C)$
for the union of $\chi(v)$ over all vertices $v \in C$. Note that if $C' \subseteq C$ is a sub-complex of a
chromatic complex, it inherits a chromatic structure by restriction.

In particular, the standard $n$-simplex $\s$ is a chromatic complex, with $\chi$ being the identity.

Every chromatic complex $C$ has a {\em standard chromatic subdivision} $\Chr C$. Let us first define $\Chr \s$
for the standard simplex $\s$. The vertices of $\Chr \s$ are pairs $(i, \t)$, where $i \in \{0,1 ,\dots, n\}$
and $\t$ is a face of $\s$ containing $i$. We let $\chi(i, \t) = i$. Further, $\Chr s$ is characterized by its
$n$-simplices; these are the $(n+1)$-tuples $((0,\t_0), \dots, (n, \t_n))$ such that:
\begin{enumerate}[(a)]
\item For all $\t_i$ and $\t_j$, one is a face of the other;
\item If $j \in \t_i$, then $\t_j \subseteq \t_i$. 
\end{enumerate} 
The geometric realization of $\s$ can be taken to be the set $\{\x=(x_0, \dots, x_n) \in [0,1]^{n+1} \mid \sum
x_i = 1\},$ where the vertex $i$ corresponding to the point $\x^i$ with $i$ coordinate $1$ and all others
coordinate $0$. Then, we can identify a vertex $(i, \t)$ of $\Chr \s$ with the point

\begin{small}\[
\frac{1}{2k-1} \x_i + \frac{2}{2k-1} \Bigl( \sum_{\{j \in \t \mid j \neq i\}} \x_j \Bigr) \ \in |\s| \subset
\R^{n+1},
\]\end{small}

where $k$ is the cardinality of $\t$. 
Thus, $\Chr \s$ becomes a subdivision of $\s$ and the geometric realizations are identical: $|\s|=|\Chr \s|$. The standard chromatic subdivision, $\Chr\s$, is illustrated for a 3-process system in Figure~\ref{Fig:scs}(a).

Next, given a chromatic complex $C$, we let $\Chr C$ be the subdivision of $C$ obtained by replacing each
simplex in $C$ with its chromatic subdivision. Thus, the vertices of $\Chr C$ are pairs $(p, \sigma)$, where $p$
is a vertex of $C$ and $\sigma$ is a simplex of $C$ containing $p$. If we iterate this process  $m$ times we
obtain the $m\th$ chromatic subdivision, $\Chr^m C$.

Let $A$ and $B$ be chromatic complexes. A simplicial map $f: A \to B$ is called a {\em chromatic map} if for all
vertices $v \in A$, we have $\chi(v) = \chi(f(v))$. Note that a chromatic map is automatically non-collapsing. A
chromatic map has chromatic subdivisions $\Chr^m f: \Chr^m A \to \Chr^m B$. Under the identifications of
topological spaces $|A| \cong |\Chr^m A|, |B| \cong |\Chr^m B|,$ the continuous maps $|f|$ and $|\Chr^m f|$ are identical.

A simplicial map $\phi$ is carried by the carrier
map $\Delta$ if $\phi(\sigma) \subset \Delta(\sigma)$ for every simplex $\sigma$ in their
domain.

\end{document}